\documentclass[12pt]{article}

\usepackage{amsmath}
\usepackage{amsfonts}
\usepackage{amssymb}
\usepackage{amsthm}
\usepackage{graphicx}
\usepackage{gauss}
\usepackage{enumitem}
\usepackage{hyperref}
\usepackage{caption,subcaption}
\usepackage{color}

\newtheorem{theorem}{Theorem}[section]
\newtheorem{lemma}[theorem]{Lemma}
\newtheorem{definition}[theorem]{Definition}
\newtheorem{corollary}[theorem]{Corollary}
\newtheorem{proposition}[theorem]{Proposition}
\newtheorem{remark}[theorem]{Remark}
\newtheorem{example}[theorem]{Example}

\numberwithin{equation}{section}

\newcommand{\Section}[1]{\section{#1} \setcounter{theorem}{0}}

\newcommand{\ignore}[1]{}
\newcommand{\etal}{et al.}

\newcommand{\Matrix}[1]{\ensuremath{\left[\begin{array}{rrrrrrrrrrrrrrrrrr} #1 \end{array}\right]}}
\newcommand{\Matrixc}[1]{\ensuremath{\left[\begin{array}{cccccccccccc|ccc} #1 \end{array}\right]}}

\newcommand{\R}{{\mathbf R}}

\renewcommand{\AA}{{\mathcal A}}
\newcommand{\BB}{{\mathcal B}}
\newcommand{\CC}{{\mathcal C}}
\newcommand{\DD}{{\mathcal D}}
\newcommand{\EE}{{\mathcal E}}

\newcommand{\GG}{{\mathcal G}}

\newcommand{\II}{{\mathcal I}}
\newcommand{\KK}{{\mathcal K}}

\newcommand{\LL}{{\mathcal L}}

\newcommand{\sS}{{\mathcal S}}

\newcommand{\WW}{{\mathcal W}} %

\renewcommand{\epsilon}{\varepsilon}

\newcommand{\PH}{H}

\newlength\savedwidth

\newcommand\thickhline{\noalign{\global\savedwidth\arrayrulewidth\global\arrayrulewidth 2pt}%
	\hline
	\noalign{\global\arrayrulewidth\savedwidth}}

\DeclareGraphicsExtensions{.pdf}
      
\title{The Structure of Infinitesimal Homeostasis in Input-Output Networks}

\author{
Yangyang Wang \\
Department of Mathematics \\
The University of Iowa \\
Iowa City, IA 52242, USA \\
\href{yangyang-wang@uiowa.edu}{yangyang-wang@uiowa.edu} \\
\and
Zhengyuan Huang \\
The Ohio State University \\
Columbus, OH 43210, USA \\
\href{huang.3224@buckeyemail.osu.edu}{huang.3224@buckeyemail.osu.edu} \\
\and
Fernando Antoneli \\
Escola Paulista de Medicina \\
Universidade Federal de S\~ao Paulo \\
S\~ao Paulo, SP 04039-032, Brazil \\
\href{fernando.antoneli@unifesp.br}{fernando.antoneli@unifesp.br} \\
\and
Martin Golubitsky \\
Department of Mathematics\\
The Ohio State University \\
Columbus, OH 43210, USA \\
\href{golubitsky.4@osu.edu}{golubitsky.4@osu.edu}
}

\date{\today}

\begin{document}

\maketitle

\begin{abstract}

Homeostasis refers to a phenomenon whereby the output $x_o$ of a system is approximately constant on variation of an input $\II$.  Homeostasis occurs frequently in biochemical networks and in other networks 
of interacting elements where mathematical models are based on differential equations associated to the network. These networks can be abstracted as digraphs $\GG$ with a distinguished input node $\iota$, a different distinguished output node $o$, and a number of regulatory nodes $\rho_1,\ldots,\rho_n$.  In these models the input-output map $x_o(\II)$ is defined by a stable equilibrium $X_0$ at $\II_0$.  Stability implies that there is a stable equilibrium $X(\II)$ for each $\II$ near $\II_0$ and infinitesimal homeostasis occurs at $\II_0$ when $(dx_o/d\II)(\II_0) = 0$.  We show that there is an $(n+1)\times(n+1)$ {\em homeostasis matrix} $H(\II)$ for which $dx_o/d\II = 0$ if and only if $\det(H) = 0$.  We note that the entries in $H$ are linearized couplings and $\det(H)$ is a homogeneous polynomial of degree $n+1$ in these entries. We use combinatorial matrix theory to factor the polynomial $\det(H)$ and thereby determine a menu of different types of possible homeostasis associated with each digraph $\GG$. Specifically, we prove that each factor corresponds to a subnetwork of $\GG$.  The factors divide into two combinatorially defined classes: {\em structural} and {\em appendage}.  Structural factors correspond to {\em feedforward} motifs and appendage factors correspond to {\em feedback} motifs.  Finally, we discover an algorithm for determining the homeostasis subnetwork motif corresponding to each factor of $\det(H)$ without performing numerical simulations on model equations. The algorithm allows us to classify low degree factors of $\det(H)$.  There are two types of degree 1 homeostasis (negative feedback loops and kinetic or Haldane motifs) and there are two types of degree 2 homeostasis (feedforward loops and a degree two appendage motif).

\end{abstract}

\setcounter{tocdepth}{2}
\tableofcontents
\enlargethispage{5mm}

\Section{Introduction}

\subsection{Overview and perspective}
\label{SS:OP}

This paper divides into three parts.  Part I, which is just Section~\ref{SS:OP}, puts our work in perspective. Part II, which consists of Sections~\ref{sec:input-output}-\ref{SS:C}, gives a precise technical description of our results.  Finally, Part III consists of Sections~\ref{sec:core}-\ref{S:CC} and contains the rigorous mathematics, along with the proofs of theorems mentioned in Part II.  We note that certain graph theoretic notions and theorems are needed in the proofs in Part III, but are not needed in the description of our results in Part II.

A system exhibits {\em homeostasis} if on change of an input variable $\II$ some observable $x_o(\II)$ remains approximately constant.  Many researchers have emphasized that homeostasis is an important phenomenon in biology.  For example, the extensive work of Nijhout, Reed, Best and collaborators~\cite{NRBU04, RLN10, BNR09, NR14, NBR15, NBR18} consider biochemical networks associated with metabolic signaling pathways. Further examples include regulation of cell number and size \cite{L13}, control of sleep~\cite{WRCD99}, and expression level regulation in housekeeping genes~\cite{AGS18}.  

{\em Adaptation} is a closely related notion.  It is the ability of a system to reset an observable $x_o(\II)$ to its prestimulated output level (its {\em set point}) after responding to an external stimulus $\II$. Adaptation has been widely used in synthetic biology and control engineering (cf. \cite{MTELT09, AM13, TM16, F16, QD18, AL18, DQMS18, ALGBSK19}).  Here, the focus of the research is on the stronger condition of {\em perfect adaptation}, where the observable $x_o(\II)$ is required to be constant over a range of external stimuli $\II$.  The literature is huge, and these articles are a small sample.

The mathematical formulation of both homeostasis and adaptation is as follows. Start with a system of ordinary differential equations usually associated to a network of interacting elements. Next define an {\em input-output function} that maps the input variable or the external stimulus $\II$  to the output $x_o(\II)$.  Then the occurrence of homeostasis or perfect adaptation is a question about the properties of $x_o(\II)$ under (time-dependent) variation of $\II$.

For instance, Reed \etal~\cite{RBGSN17} consider biochemical signaling networks whose nodes represent the concentrations of certain biochemical substrates that interact through mass action kinetics.  They identify two homeostasis motifs in three-node networks: the {\em feedforward loop} motif (FFL)  (Figure \ref{E:FFL}) and the {\em kinetic} motif (K) (Figure~\ref{E:haldane}). There is notation in these figures that we have not yet defined.  In related work on three-node biochemical networks with Michaelis-Menten kinetics, Ma \etal~\cite{MTELT09} identify numerically two network topologies that achieve perfect adaptation.  To do this, the authors searched 16,038 equations in various three-node network topologies over a wide range of parameter space. They found just two motifs that achieved perfect adaptation: the {\em negative feedback loop} motif (NFL) (Figure \ref{E:ND}) and the {\em incoherent feedforward loop} (IFL) (Figure \ref{E:FFL}).  The combined results of \cite{RBGSN17} and \cite{MTELT09} show that at least three network topologies (K, NFL, IFL $\cong$ FFL) emerge as motifs exhibiting homeostasis or perfect adaptation in three-node biochemical networks.  

Recently, Wang and Golubitsky~\cite{GW19} classified the `homeostasis types' that can occur in three-node input-output networks based on the notion of {\em infinitesimal homeostasis} \cite{GS17} (see Definition \ref{D:homeostasis}).  Using this approach, they were able to reproduce the classification results in \cite{MTELT09} and \cite{RBGSN17}, within a broader class of systems including, but not limited to, specific model systems based on mass action or Michaelis-Menten kinetics.  They showed that three-node networks that can exhibit infinitesimal homeostasis are, up to core equivalence (see Definition~\ref{D:backward}), the three network topologies mentioned above.  

This paper generalizes the results of \cite{GW19} on three-node networks to arbitrarily large {\em input-output networks}.  We follow \cite{GS06} and abstract the notion of biochemical network to a `math network' given by a digraph $\GG$ with a distinguished input node $\iota$ and a different distinguished output node $o$.  The specific model equations are abstracted into {\em admissible} systems of differential equations, namely, one-parameter smooth families of vector fields compatible with the network topology of $\GG$, such that only the input node depends explicitly on $\II$.  These networks and their associated systems of differential equations are called {\em input-output networks}.  We show that under certain conditions (the existence of an asymptotically stable equilibrium $X_0$ for a particular parameter value $\II_0$), one can always define the {\em input-output function} $\II\mapsto x_o(\II)$ associated to a given input-output network $\GG$.

A straightforward application of Cramer's rule (Lemma~\ref{L:det}) gives a useful method for computing infinitesimal homeostasis points: infinitesimal homeostasis occur at $\II_0$, namely, $\frac{dx_o}{d\II}(\II_0)=0$, if and only if $\det\!\big(H(\II_0)\big)=0$ (Section~\ref{sec:cramer}). This result motivates the introduction of the \emph{homeostasis matrix} $H(\II)$ (see equation \eqref{xo'_reduced2}), whose entries are linearized coupling strengths and linearized self-coupling strengths associated with the input-output network.  The homeostasis matrix $H$ -- which has appeared in the literature under different names and notations (cf.~\cite{MTELT09,AM13,TM16,GS17,AL18,ALGBSK19}) -- is the central object in our theory. 
As an aside: In our math networks arrows are identical and represent couplings and nodes are identical and represent differential equations, but neither the couplings nor the equations are assumed to be identical.

Our main result states that the homeostasis types that occur in admissible systems of differential equations associated with the network $\GG$ are classified by the topology of certain subnetwork motifs of $\GG$.  Moreover, there is an algorithm (Section~\ref{sec:algorithm}) for determining all the homeostasis subnetwork motifs and the corresponding homeostasis conditions, which also can be used for designing network topologies that display infinitesimal homeostasis.
  
In order to prove our results we introduce new concepts and techniques.  The notion of {\em core network} (Section~\ref{sec:core-main}) allows one to go from a general input-output network to a `minimal network' that retains all essential features of homeostasis.  We define {\em core equivalence} of core networks in such a way that the determinant of a homeostasis matrix is determined by its core equivalence class.  Combinatorial matrix theory~\cite{BR91} lets us put $H$ into block upper triangular form and each diagonal block $B_\eta$ is irreducible (no further triangularization is possible) and corresponds to a {\em homeostasis type} (Section~\ref{sec:homeostasis-blocks}).  The {\em degree} of the homeostasis type is defined as the size $k$ of the square block $B_\eta$ and we prove that each block $B_\eta$ has either $k$ or $k-1$ self-couplings.  In the first case we call the homeostasis type {\em appendage class} and in the second {\em structural class} (Section~\ref{SS:homeostasis_vs_networks}).  We characterize combinatorially both homeostasis types by identifying homeostasis subnetwork motifs and associating a subnetwork motif to each homeostasis type (Section~\ref{SS:topological}). We also give an algorithm that determines the homeostasis blocks and their respective homeostasis types (Section~\ref{sec:algorithm}).

In the biochemical network literature on homeostasis (or adaptation) it is usual to find designations attached to the networks, such as {\em negative feedback loop}, {\em antithetical integral feedback}, {\em incoherent feedforward loop}, etc. \cite{MTELT09,TM16,F16}.  These names refer to the presence of a certain mechanism that is responsible for the occurrence of homeostasis in a particular network.   Ma et al.~\cite{MTELT09} suggest that studies of these mechanisms can yield design principles for constructing network topologies that exhibit homeostasis. This could be called a `bottom-up' approach for constructing homeostasis.  It starts by identifying small building blocks that are associated with homeostasis and then how the blocks can be combined to build-up increasingly  more complex networks that exhibit homeostasis. Here we take a 'top-down' approach. We start with an input-output network $\GG$ and have an algorithm that shows us how homeostasis in $\GG$ can be generated from homeostasis in certain subnetworks.

Fundamental to our approach is the discovery that homeostasis in $\GG$ can be associated with only two `classes of mechanisms' that we called {\em structural} and {\em appendage}, each associated with certain topological properties (Section~\ref{SS:topological}).  In addition to classifying homeostasis types in a given network, these topological constraints also provide insights into the `bottom-up' construction of homeostasis systems.  The structural and appendage classes are abstract generalizations of the usual `feedforward' and `feedback' mechanisms \cite{MTELT09,F16}. More precisely, for each homeostasis type (in each class), there is a corresponding `network motif' and an associated homeostasis mechanism.  For instance, {\em negative feedback loop} and {\em antithetical integral feedback} are types in the appendage class, and {\em incoherent feedforward loop} is a type in the structural class.

The motivation for the term {\em structural homeostasis} comes from \cite{RBGSN17}, where the authors identify the feedforward loop as one of the homeostastic motifs in three-node biochemical networks. In general, structural homeostasis corresponds to a balancing of two or more excitatory/inhibitory sequence of couplings from the input node to the output node; that is, a generalized feedforward loop.  There is a degenerate case where the role of the balancing is played by neutral coupling, a transition state between excitation and inhibition.  This homeostasis type is called {\em Haldane}, because Haldane \cite{H30} seems to have been the first to observe this homeostatic mechanism.

The intuition behind the term {\em appendage homeostasis} is that homeostasis is generated by a cycle of regulatory nodes; that is, a generalized feedback loop.  This loop functions as {\em controller nodes} on a system that does not by itself exhibit homeostasis. There is a degenerate case of appendage homeostasis that we call {\em null degradation} where the role of the controller is played by a neutral node that balances between degradation and production.  See Subsection~\ref{SS:FFFB} for additional detail.

A striking outcome of our approach is that we do not need to specify any homeostasis generating mechanisms at the outset.  However, we find \emph{a posteriori} that (given the appropriate generalizations) there are essentially only the two well-known {\em feedback} / {\em feedforward} types of homeostasis generating mechanisms.

Our work is unusual in that it combines ideas from combinatorial matrix theory and graph theory adapted to input-output networks to determine properties of equilibria of differential equations.  Specifically, the \emph{determinant formula} (Theorem~\ref{L:summand_form}) connects the nonzero summands of $\det(H)$ with simple paths from the input node to the output node of the network $\GG$.  It is reminiscent of the connection between a directed graph and its adjacency matrix~\cite{BC09}.  These simple paths allow us to identify both structural and appendage homeostasis. 
Finally, our theoretical results also allow us to derive formulas for determining the {\em chair singularities} \cite{NBR14,RBGSN17}.

\subsection{Input-output networks and infinitesimal homeostasis}
\label{sec:input-output}

We now define the basic objects: {\em input-output networks}, {\em network admissible systems of differential equations}, and {\em input-output functions}.

An {\em input-output network} $\GG$ has a distinguished {\em input node} $\iota$, a distinguished {\em output node} $o$ (distinct from $\iota$), and $n$ {\em regulatory nodes} $\rho = (\rho_1,\ldots,\rho_n)$.  The network $\GG$ also has a specified set of arrows (or directed edges) connecting nodes $\ell$ to nodes $j$.  The associated network systems of differential equations have the form 
\begin{equation} \label{eq:ad_coord}
\begin{array}{rcl}
\dot{x}_\iota & = & f_\iota(x_\iota,x_\rho, x_o,\II) \\
\dot{x}_\rho & = & f_\rho(x_\iota,x_\rho, x_o) \\
\dot{x}_o & = & f_o(x_\iota,x_\rho, x_o) 
\end{array}
\end{equation}
where $\II\in\R$ is an {\em external input parameter}, $X=(x_\iota,x_{\rho},x_o)\in\R\times\R^n\times\R$ is the vector of state variables associated to the network nodes and $F(X,\II)=(f_\iota, f_{\rho}, f_o)$ is a smooth one-parameter family of {\em $\GG$-admissible vector fields} on the state space $\R\times\R^n\times\R$ (see~\cite{GS06} for the definition of the space of admissible vector fields attached to a given network $\GG$).
We write the network system \eqref{eq:ad_coord} as
\begin{equation}\label{eq:ad}
\dot{X} = F(X, \II)
\end{equation}

Let $f_{j,x_\ell}$ denote the partial derivative of the $j^{th}$ node function $f_j$ with respect to the $\ell^{th}$ node variable $x_\ell$.  We make the following assumptions about the vector field $F$ throughout:
\begin{enumerate}[label=(\alph*)]
\item $F$ has an asymptotically stable equilibrium at $(X_0,\II_0)$.
\item The partial derivative $f_{j,x_\ell}$ can be nonzero only if the network $\GG$ has an arrow $\ell\to j$.
\item Only the input node coordinate function $f_\iota$ depends on the external input parameter $\II$ and the partial derivative of $f_\iota$ with respect to $\II$ at the equilibrium point $(X_0,\II_0)$ satisfies
\begin{equation} \label{e:f_iota_I}
f_{\iota,\II}(X_0,\II_0) \neq 0
\end{equation}
\end{enumerate}

It follows from (a) and the implicit function theorem applied to
\begin{equation} \label{eq:F=0}
F(X,\II) = 0
\end{equation}
that there exists a unique smooth family of stable equilibria
\begin{equation}
X(\II) = \big(x_\iota(\II),x_\rho(\II),x_o(\II)\big)
\end{equation}
such that $F(X(\II),\II)\equiv 0$ and $X(\II_0) = X_0$.  

\begin{definition} \rm \label{D:IOF}
The mapping $\II \mapsto x_o(\II)$ is called the {\em input-output function}.
\end{definition}

Local homeostasis is defined near $\II_0$ when the input-output function $x_o$ is approximately constant near $\II_0$.  An important observation is that locally homeostasis occurs when the derivative of $x_o$ with respect to $\II$ is zero at $\II_0$.  More precisely: 

\begin{definition} \rm \label{D:homeostasis}
{\em Infinitesimal homeostasis} occurs at $\II_0$ if $x'_o(\II_0) = 0$ where $'$ indicates differentiation with respect to $\II$. 
\end{definition}

Terms that involve coupling in network systems are:
\begin{definition} \rm
Let $F=(f_\iota, f_{\rho}, f_o)$ be an admissible system for the network $\GG$.
\begin{enumerate}[label=(\alph*)]
\item The partial derivative $f_{j,x_\ell}(X_0,\II_0)$ is the linearized {\em coupling} associated with the arrow $\ell\to j$ at the equilibrium $(X_0,\II_0)$.
\item The partial derivative $f_{j,x_j}(X_0,\II_0)$ is the linearized {\em self-coupling} of node $j$ at the equilibrium $(X_0,\II_0)$.
\end{enumerate}
\end{definition}


\begin{remark} \rm
A notion similar to infinitesimal homeostasis, called {\em perfect homeostasis} or {\em perfect adaptation}, requires the stronger condition that the derivative of the input-output function be identically zero on an interval.  It follows from Taylor's theorem that infinitesimal homeostasis implies that the input-output function $x_o$ is approximately constant near $\II_0$, the converse is not valid in general~\cite{RBGSN17}.  This property is called {\em near perfect homeostasis} or {\em near perfect adaptation} in the literature (cf. \cite{F16,TM16}).  Hence, infinitesimal homeostasis is an intermediate notion between {\em perfect homeostasis} and {\em near perfect homeostasis}.
\end{remark}

\subsection{Infinitesimal homeostasis using Cramer's rule}
\label{sec:cramer}

As noted previously \cite{GS17,RBGSN17,GW19}, a straightforward application of Cramer's rule gives a formula for determining  infinitesimal homeostasis points.  See Lemma~\ref{L:det}.

We use the following notation.  Let $J$ be the $(n+2)\times(n+2)$ Jacobian matrix of \eqref{eq:ad} and let $\PH$ be the $(n+1)\times(n+1)$ {\em homeostasis matrix} given by dropping the first row and the last column of $J$:
\begin{equation} \label{xo'_reduced2}
J = \Matrixc{ f_{\iota, x_\iota}   &  f_{\iota, x_\rho} & f_{\iota, x_o} \\
	f_{\rho, x_\iota}   &  f_{\rho, x_\rho} & f_{\rho, x_o} \\
	f_{o, x_\iota} &  f_{o, x_\rho} & f_{o, x_o} }
\quad 
H = \Matrixc{ f_{\rho, x_\iota}   &  f_{\rho, x_\rho}\\ f_{o, x_\iota} &  f_{o, x_\rho}}
\end{equation}
Here all partial derivatives $f_{\ell,x_j}$ are evaluated at the equilibrium $X_0$. The next lemma has appeared in several places including \cite[Figure~5A]{MTELT09}, \cite[Lemma~6.1]{GS17}, \cite[Theorem~3]{RBGSN17},and \cite[Lemma~9]{GW19}, though not in this generality.  The proof is included here for completeness, even though it is virtually identical to the one in \cite{GW19}.

\begin{lemma} \label{L:det}
Let $(X_0,\II_0)$ be an asymptotically stable equilibrium of \eqref{eq:ad}.  The input-output function $x_o(\II)$ satisfies
\begin{equation} \label{xo'}
x_o' = \pm\frac{f_{\iota, \II}}{\det(J)} \det(\PH)
\end{equation}
Hence, $\II_0$ is a point of infinitesimal homeostasis if and only if 
\begin{equation} \label{xo'_reduced}
\det(\PH)= 0
\end{equation}
at $(X_0, \II_0)$.
\end{lemma}

\begin{proof}
Implicit differentiation of \eqref{eq:F=0} with respect to $\II$ yields the matrix system
\begin{equation} \label{imp_diff}
J \Matrixc{x_i' \\ x_\rho' \\ x_o'} = -\Matrixc{f_{\iota, \II} \\ 0 \\ 0}
\end{equation}
Since $X_0$ is assumed to be a stable equilibrium, it follows that $\det(J)\neq 0$. 
On applying Cramer's rule to \eqref{imp_diff} we can solve for $x_o'$ obtaining  
\begin{equation}
x_o'(\II_0) = \frac{1}{\det(J)}\det\Matrixc{
	f_{\iota, x_\iota} &  f_{\iota, x_\rho} & -f_{\iota, \II} \\ 
	f_{\rho, x_\iota}&  f_{\rho, x_\rho} & 0 \\
	f_{o, x_\iota} &  f_{o, x_\rho} & 0 }
\end{equation}
which leads to \eqref{xo'}.
By assumption \eqref{e:f_iota_I}, $f_{\iota, \II} \neq 0$.  Hence, the fact that infinitesimal homeostasis for \eqref{eq:ad} is equivalent to \eqref{xo'_reduced} follows directly from \eqref{xo'}.
\qed
\end{proof}

\subsection{Core networks}
\label{sec:core-main}

Homeostasis in a given network $\GG$ can be determined by analyzing a simpler network that is obtained by eliminating certain nodes and arrows from $\GG$.  We call the network formed by the remaining nodes and arrows the {\em core subnetwork}.

\begin{definition} \rm \label{D:updown}
A node $\tau$ in a network $\GG$ is \emph{downstream} from a node $\rho$ in $\GG$ if there exists a path in $\GG$ from $\rho$ to $\tau$.  Node $\rho$ is \emph{upstream} from node $\tau$ if $\tau$ is downstream from $\rho$. 
\end{definition}

These relationships are important when trying to classify infinitesimal homeostasis.  For example, if the output node $o$ is not downstream from the input node $\iota$, then the input-output function $x_o(\II)$ is identically constant in $\II$.  Although technically this is a form of infinitesimal homeostasis, it is an uninteresting form.

\begin{definition} \rm \label{D:core}
\begin{enumerate}[label=(\alph*)]
\item The input-output network is a {\em core network} if every node is both upstream from the output node and downstream from the input node.  
\item Every input-output network $\GG$ has a core subnetwork $\GG_c$ whose nodes are the nodes in $\GG$ that are both upstream from the output node and downstream from the input node and whose arrows are the arrows in $\GG$ whose head and tail nodes are both nodes in $\GG_c$. 
\end{enumerate}
\end{definition}

\begin{figure}[!ht]
\centering
\begin{subfigure}{0.45\textwidth}
\centering
\includegraphics[width=1\textwidth]{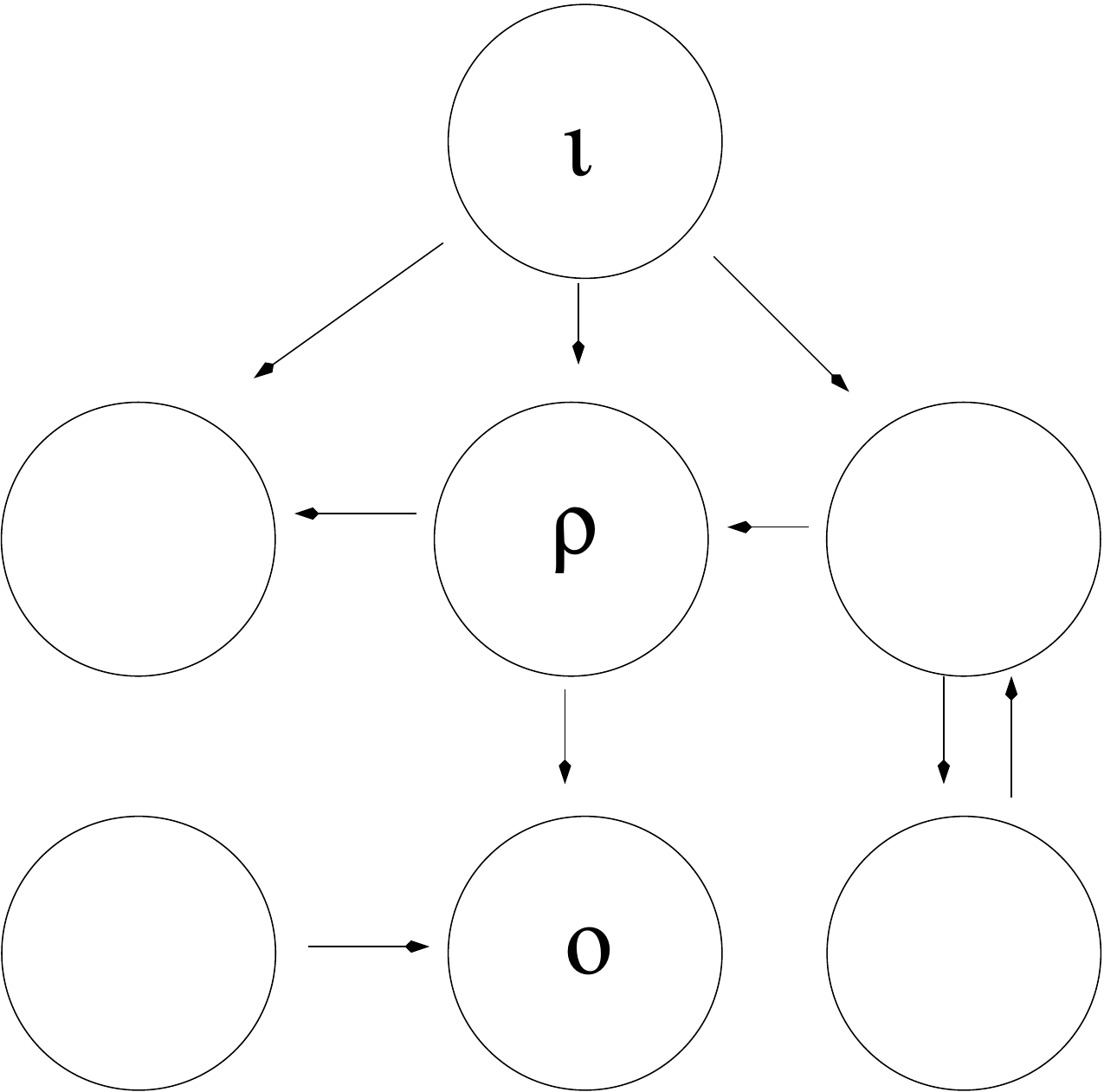}
\caption{Network $\GG$ \label{7nodeG}}
\end{subfigure} 
\qquad\quad
\begin{subfigure}{0.45\textwidth}
\centering
\includegraphics[width=0.6\textwidth]{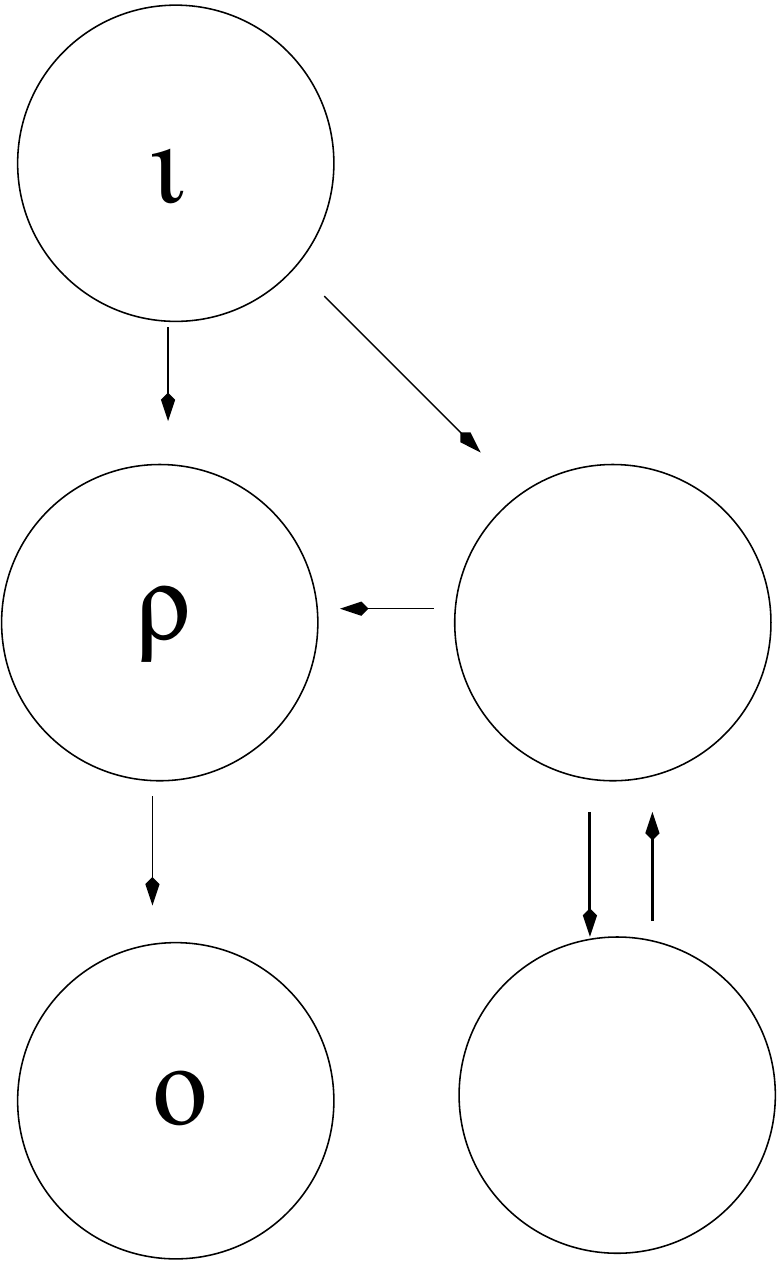}
\caption{Core network $\GG_c$ \label{7nodeGC}}
\end{subfigure} 
\caption{Example of a core subnetwork. $\GG_c$ obtained from $\GG$ by deleting nodes that are not both upstream from $o$ and downstream from $\iota$ and all arrows into and out of the deleted nodes. \label{F:7node}}
\end{figure}

The next result concerning core networks follows from Theorem~\ref{T:core}.

\begin{corollary} \label{T:coreA}
Let $\GG$ be an input-output network and let $\GG_c$ be the associated core subnetwork. The input-output function associated with $\GG_c$ has a point of infinitesimal homeostasis at $\II_0$ if and only if the input-output function associated with $\GG$ has a point of infinitesimal homeostasis at $\II_0$.
\end{corollary}

It follows from Corollary~\ref{T:coreA} that classifying infinitesimal homeostasis for networks $\GG$ is equivalent to classifying infinitesimal homeostasis for the core subnetwork $\GG_c$.

Figure~\ref{F:7node} gives an example of reducing a network to a core network.  In this case the left two nodes in Figure~\ref{7nodeG} are deleted to get to the core network, which is illustrated in Figure~\ref{7nodeGC}.

\begin{definition} \rm \label{D:backward}
\begin{enumerate}[label=(\alph*)]
\item Two $(n+2)$-node core networks are \emph{core equivalent} if the determinants of their homeostasis matrices are identical polynomials of degree $n+1$.
\item A {\em backward arrow} is an arrow whose head is the input node $\iota$ or whose tail is the output node $o$.
\end{enumerate}
\end{definition}

\begin{corollary} \label{P:core_equivalent}
If two core networks differ from each other by the presence or absence of backward arrows, then the core networks are core equivalent.
\end{corollary} 

\begin{proof}
The linearized couplings associated to backward arrows are of form $f_{\iota,x_k}$ and $f_{k,x_o}$, which do not appear in the homeostasis matrix~\eqref{xo'_reduced}. \qed
\end{proof}


Therefore, backward arrows can be ignored when computing infinitesimal homeostasis with the homeostasis matrix $\PH$.  However, backward arrows cannot be totally ignored, since they are involved in the determination of both the equilibria of \eqref{eq:ad} and their stability.

Corollary~\ref{P:core_equivalent} can be generalized to a theorem giving necessary and sufficient graph theoretic conditions for core equivalence.  See Theorem~\ref{T:core_equivalent}. 

\subsection{Infinitesimal homeostasis blocks}
\label{sec:homeostasis-blocks}

The previous results imply that the computation of infinitesimal homeostasis reduces to solving $\det(\PH)=0$, where $\PH$ is the homeostasis matrix associated with a core network.  From now on we assume that the input-output network $\GG$ is a core network.

It is important to observe that the nonzero entries of $\PH$ are the linearized coupling strengths $f_{j,x_\ell}$ for the network connected nodes $\ell\to j$ and the linearized self-coupling strengths $f_{j,x_j}$.  It follows that $h = \det(\PH)$ is a homogeneous polynomial of degree $n+1$ in the $(n+1)^2$ entries of $\PH$.   We use combinatorial matrix theory to show that in general $h$ is nonzero and can factor, and that there is a different type of infinitesimal homeostasis associated with each factor.
(Note that if $h\equiv 0$, then the input-output function is constant.)


Frobenius-K\"onig theory~\cite{BR91} (see~\cite{S77} for an historical account) applied to the homeostasis matrix $\PH$ implies that there are two constant $(n+1)\times(n+1)$ permutation matrices $P$ and $Q$ such that
\begin{equation} \label{eq:FK_form}
P \PH Q = \Matrixc{
	B_1 & *   & \cdots & * \\
	0   & B_2 & \cdots & * \\
	\vdots &  &  & \vdots \\
	0   &  0 & \cdots & B_m
}   
\end{equation}
where the square matrices $B_1,\ldots,B_m$ are unique up to permutation.  More precisely, each block $B_\eta$ cannot be brought into the form \eqref{eq:FK_form} by permutation of its rows and columns. Hence
\begin{equation} \label{eq:FK_factors}
\det(\PH) = \det(B_1) \cdots \det(B_m) \quad\mbox{or}\quad h = h_1\cdots h_m
\end{equation}
is a unique factorization since $h_\eta = \det(B_\eta)$ cannot further factor for each $\eta$; that is, each $\det(B_\eta)$ is an irreducible homogeneous polynomial.  Specifically:

\begin{theorem} \label{T:block}
The polynomial $h_\eta=\det(B_\eta)$ is {\em irreducible} (in the sense that it cannot be factored as a polynomial) if and only if the block submatrix $B_\eta$ is {\em irreducible} (in the sense that $B_\eta$ cannot be brought to the form \eqref{eq:FK_form} by permutation of rows and columns of $B_\eta$). 
\end{theorem}

\begin{proof} 
The decomposition \eqref{eq:FK_form} corresponds to the irreducible components in the factorization \eqref{eq:FK_factors} follows from \cite[Theorem 4.2.6 (pp.~114--115) and Theorem 9.2.4 (p.~296)]{BR91}.
\qed
\end{proof}

A consequence of \eqref{eq:FK_factors} and \eqref{xo'_reduced} is that for each $\eta = 1,\ldots, m$ there is a {\em defining condition} for infinitesimal homeostasis given by the polynomial equation $\det(B_\eta) = 0$.  Recall that the input-output function is implicitly defined in terms of the external input $\II$ and $\det(B_\eta)$ is a homogeneous polynomial in the linearized coupling strengths $f_{j,x_\ell}$ evaluated at $X(\II)$.  Hence, there are $m$ different defining conditions for infinitesimal homeostasis, $h_\eta(\II) = 0$, where each one  gives a nonlinear equation that can be solved for some $\II=\II_0$.  In practice, for a given model, it is unlikely that these equations can be solved in closed form;  however, it is possible that each defining condition can be solved numerically.  So, the decomposition of the homeostasis matrix $\PH$ into blocks $B_\eta$ simplifies the solution of $\det(\PH)=0$. 

\begin{definition} \label{D:homo_type} \rm
Given the homeostasis matrix $\PH$ of an input-output network $\GG$, we call the unique irreducible diagonal blocks $B_\eta$ in the decomposition \eqref{eq:FK_form} {\em irreducible components}.  We say that homeostasis in $\GG$ is {\em of type $B_\eta$} if $\det(B_\eta) = 0$ and $\det(B_\xi) \neq 0$ for all $\xi\neq\eta$. 
\end{definition}

\subsection{Infinitesimal homeostasis classes}
\label{SS:homeostasis_vs_networks}

The next results assert that the irreducible components $B_\eta$ of $\PH$ determine two distinct {\em homeostasis classes} (appendage and structural) and that one can associate a subnetwork $\KK_\eta$ of $\GG$ with each $B_\eta$ (see Section~\ref{S:classification2}).

Let $B_\eta$ be an irreducible component in the decomposition \eqref{eq:FK_form}, where $B_\eta$ is a $k\times k$ diagonal block, that is, $B_\eta$ has degree $k$.  Since the entries of $B_\eta$ are entries of $\PH$, these entries have the form $f_{\rho,x_\tau}$; that is, the entries are either $0$ (if $\tau\to\rho$ is not an arrow in $\GG$), {\em self-coupling} (if $\tau = \rho$), or {\em coupling} (if $\tau\to\rho$ is an arrow in $\GG$).  

Since $P$ and $Q$ in \eqref{eq:FK_form} are constant permutation matrices, all entries in each row (resp.~column) of $B_\eta$ must lie in a single row (resp.~column) of $\PH$.  Hence, $B_\eta$ has the form
\begin{equation} \label{eq:K}
B_\eta = \Matrixc{f_{\rho_1,x_{\tau_1}} & \cdots & f_{\rho_1,x_{\tau_k}} \\ 
	\vdots & \ddots & \vdots \\
	f_{\rho_k,x_{\tau_1}} & \cdots & f_{\rho_k,x_{\tau_k}}}
\end{equation}
It follows that the number of self-coupling entries of $B_\eta$ are the same no matter which permutation matrices $P$ and $Q$ are used in \eqref{eq:FK_form} to determine $B_\eta$.  In Theorem~\ref{L:self_couplings} we show that a $k\times k$ submatrix $B_\eta$ has either $k$ or $k-1$ self-coupling entries.  

\begin{definition} \rm \label{D:classes}
The homeostasis class of an irreducible component $B_\eta$ of degree $k$ is {\em appendage} if $B_\eta$ has $k$ self-couplings and {\em structural} if $B_\eta$ has $k-1$ self-couplings.
\end{definition}

\begin{definition} \rm \label{D:K_eta}
The subnetwork $\KK_\eta$ of $\GG$ associated with the homeostasis block $B_\eta$ is defined as follows.  The {\em nodes} in $\KK_\eta$ are the union of nodes $p$ and $q$ where $f_{p,x_q}$ is a nonzero entry in $B_\eta$ and the {\em arrows} of $\KK_\eta$ are the union of arrows $q\to p$ where $p\neq q$. 
\end{definition}

Theorem~\ref{lem:associated_network} implies that when $B_\eta$ is appendage, the subnetwork $\KK_\eta$ has $k$ nodes and $B_\eta$ can be put in a \emph{standard Jacobian form} without any distinguished nodes (\eqref{eq:K_matrices(k)}). Also, when $B_\eta$ is structural, the subnetwork $\KK_\eta$ has $k+1$ nodes and $B_\eta$ can be put in a \emph{standard homeostasis form} with designated input node and output node (\eqref{eq:K_matrices(k-1)}).  Moreover, in this case, the subnetwork $\KK_\eta$ has no backward arrows.  That is, $\KK_\eta$ has no arrows whose head is the input node or whose tail is the output node.  See Remark~\ref{rm:associated_network} for details.

\subsection{Combinatorial characterization of homeostasis}
\label{SS:topological}

In Subsections~\ref{SS:SN} and \ref{ss:prop-appendage} we define a number of combinatorial terms. These terms are illustrated in the $12$-node network in Figure~\ref{F:homeostasis-example-comp}.

\subsubsection{Simple nodes}
\label{SS:SN}

Core input-output networks $\GG$ have combinatorial properties that we now define and exploit.  The key ideas are the concepts of $\iota o$-simple paths and super-simple nodes.

\begin{definition} \rm \label{D:simple_complementary}
Let $\GG$ be a core input-output network.
\begin{enumerate}[label=(\alph*)]
\item A directed path connecting nodes $\rho$ and $\tau$ is called a \emph{simple path} if it visits each node on the path at most once.  
\item \label{iota_o} An {\em $\iota o$-simple path} is a simple path connecting the input node $\iota$ to the output node $o$.
\item \label{simp_app} A node in $\GG$ is {\em simple} if the node lies on an $\iota o$-simple path and {\em appendage} if the node is not simple.
\item \label{super-simple} A {\em super-simple} node is a simple node that lies on every $\iota o$-simple path.
\end{enumerate}
\end{definition}
Nodes $\iota$ and $o$ are super-simple since by definition these nodes are on every $\iota o$-simple path. 
Lemma~\ref{L:supersimple} shows that super-simple nodes are well ordered (by downstream ordering) and hence adjacent super-simple pairs of nodes can be identified.

\subsubsection{Properties of appendage homeostasis}
\label{ss:prop-appendage}

Characterization of appendage homeostasis networks requires the following definitions.

\begin{definition} \rm \label{D:appendage_terms} 
Let $\GG$ be a core input-output network.
\begin{enumerate}[label=(\alph*)]
\item \label{appen} The {\em appendage subnetwork} $\AA_\GG$ of $\GG$ is the subnetwork consisting of all appendage nodes and all arrows in $\GG$ connecting appendage nodes. 
\item \label{compl} The {\em complementary subnetwork} of an $\iota o$-simple path $S$ is the subnetwork $\CC_S$ consisting of all nodes not on $S$ and all arrows in $\GG$ connecting those nodes.
\item \label{path_comp} Nodes $\rho_i,\rho_j$ in $\AA_\GG$ are {\em path equivalent} if there exists paths in $\AA_\GG$ from $\rho_i$ to $\rho_j$ and from $\rho_j$ to $\rho_i$.  An {\em appendage path component} (or an {\em appendage strongly connected component}) is a path equivalence class in $\AA_\GG$.
\item \label{cycle_def} A {\em cycle} is a path whose first and last nodes are identical.
\item  \label{no_cycle_def} Let $\KK\subset \AA_\GG$ be an appendage path component.  
A cycle is {\em $\KK$-bad} if it contains both nodes in $\KK$ and simple nodes, but it does not contain super-simple nodes.
$\KK$ satisfies the {\em no-cycle condition} if there are no $\KK$-bad cycles with nodes in $\KK$.
\end{enumerate}
\end{definition}

Note that the definition of the {\em no-cycle condition in Definition~\ref{D:appendage_terms}\ref{no_cycle_def} is equivalent to: Let $\KK\subset \AA_\GG$ be an appendage path component.  We say that $\KK$ satisfies the {\em no cycle condition} if for every $\iota o$-simple path $S$, nodes in $\KK$ do not form a cycle with nodes in $\CC_S\setminus\KK$.}

In Section~\ref{S:appendage_blocks} we prove that every subnetwork $\KK_\eta$ of $\GG$ associated with an irreducible appendage homeostasis block $B_\eta$ consists of appendage nodes (Lemma~\ref{T:KK_appendage}), is an appendage path component of $\AA_\GG$, and satisfies the no cycle condition (Theorem~\ref{thm:appendage}). The converse is proved in Theorem~\ref{T:appendage_blocks}.

\begin{remark} \rm \label{R:AiBi}
Nodes in the appendage subnetwork $\AA_\GG$ can be written uniquely as the disjoint union 
\begin{equation} \label{e:AiBi}
\AA_\GG = (\AA_1\dot{\cup}\cdots\dot{\cup}\AA_s)\; \dot{\cup}\;(\BB_1\dot{\cup}\cdots\dot{\cup}\BB_t)
\end{equation}
where each $\AA_i$ is an appendage path component that satisfies the no cycle condition and each $\BB_i$ is an appendage path component that violates the no cycle condition. Moreover, each $\AA_i$ (resp. $\BB_i$) can be viewed as a subnetwork of $\AA_\GG$ by including the arrows in $\AA_\GG$ that connect nodes in $\AA_i$ (resp. $\BB_i$).  We call $\AA_i$ a {\em no cycle appendage path component} and $\BB_i$ a {\em cycle appendage path component}.
\end{remark}

\subsubsection{A $12$-node example illustrating combinatorial terms}
\label{SSS:example}

\begin{figure}[!ht]
\centering
\includegraphics[width=.75\linewidth]{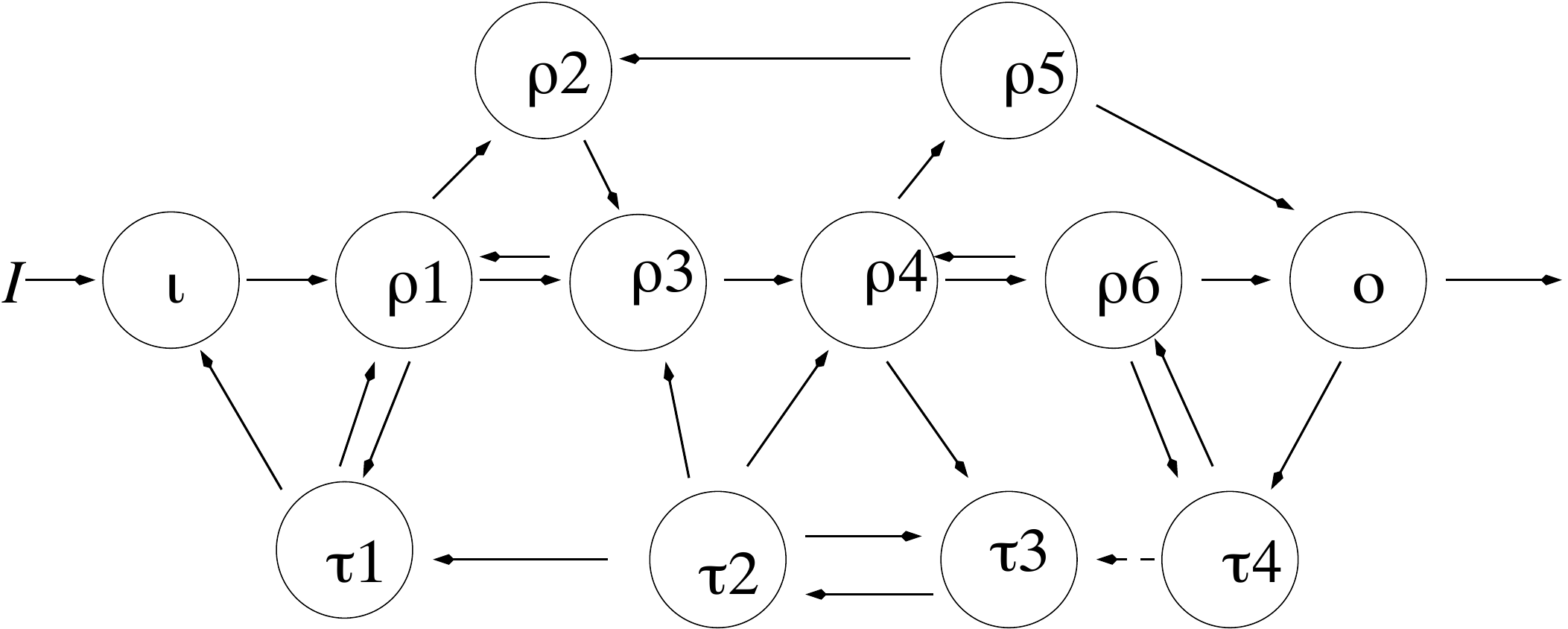}
\caption{{\bf The 12-node example.} \label{F:homeostasis-example-comp}} 
\end{figure}

The 12 nodes consist of the input node $\iota$, the output node $o$, six simple nodes $\rho_1,\ldots,\rho_6$, and four appendage nodes $\tau_1,\ldots,\tau_4$.  See Definition~\ref{D:simple_complementary}\ref{simp_app}. The network has four $\iota o$-simple paths (see Definition~\ref{D:simple_complementary}\ref{iota_o} and Table~\ref{T:iota_o_simple}) and five super-simple nodes $\iota,\rho_1,\rho_3,\rho_4,o$. See Definition~\ref{D:simple_complementary}\ref{super-simple}. 

\begin{table}[!h]
\begin{center}
\caption{{\bf Four $\iota o$-simple paths for network in Figure~\ref{F:homeostasis-example-comp}} \label{T:iota_o_simple}}
\begin{tabular}{|c|c|} \hline
Simple path ($S$) & Complementary subnetwork ($\CC_{S}$) \\ \thickhline
$\iota\to\rho_1\to\rho_2\to\rho_3\to\rho_4\to \rho_5\to o$ &
$\{\tau_1,\,\tau_2, \tau_3,\,\tau_4, \rho_6\}$ \\ \hline
$\iota \to \rho_1\to\rho_2\to\rho_3\to\rho_4\to \rho_6\to o$ &
$\{\tau_1, \tau_2,\tau_3,\,\tau_4,\,\rho_5\}$ \\ \hline
$\iota \to\rho_1\to\rho_3 \to\rho_4\to \rho_5\to o$ &
$\{\tau_1,\,\rho_2,\tau_2, \tau_3,\,\tau_4, \rho_6\}$ \\ \hline
$\iota \to \rho_1\to\rho_3 \to\rho_4\to \rho_6\to o$ &
$\{\tau_1,\tau_2,\tau_3, \,\tau_4,\,\rho_5,\rho_2\}$ \\ \hline
\end{tabular}
\end{center}
\end{table}

The appendage subnetwork is $\tau_1\leftarrow\tau_2\leftrightarrow\tau_3\leftarrow\tau_4$. See Definition~\ref{D:appendage_terms}\ref{appen}. The complementary subnetworks (see Definition~\ref{D:appendage_terms}\ref{compl}) are listed in Table~\ref{T:iota_o_simple}. There are three path components in the appendage subnetwork, namely, $\tau_1$, $\tau_2\leftrightarrow\tau_3$, and $\tau_4$.  See Definition~\ref{D:appendage_terms}\ref{path_comp}. 

An example of a cycle containing both appendage and simple nodes is $\tau_4\to\rho_6\to o\to\tau_4$.  See Definition~\ref{D:appendage_terms}\ref{cycle_def}.  Note that this cycle is $\tau_4$-bad. The two path components $\tau_1$ and $\tau_2\leftrightarrow\tau_3$ satisfy the no-cycle condition given in Definition~\ref{D:appendage_terms}\ref{no_cycle_def} since there are no $\KK$-bad cycles for $\KK=\{\tau_1\}$ or $\KK = \{\tau_2\leftrightarrow\tau_3\}$.

\subsubsection{Properties of structural homeostasis}
\label{ss:prop-structural}

Corollary \ref{C:adjacent} shows that if $B_\eta$ corresponds to an irreducible structural block, then $\KK_\eta$ has two adjacent super-simple nodes (Proposition~\ref{T:structural_supersimple_two}) and these super-simple nodes are the input node $\ell$ and the output node $j$ in $\KK_\eta$.  In addition, it follows from the standard homeostasis form (Theorem~\ref{lem:associated_network}) that the network $\KK_\eta$ contains no backward arrows. That is, no arrows of $\KK_\eta$ go into the input node $\ell$ nor out of the output node $j$.

We use the properties of structural homeostasis to construct all structural homeostasis subnetworks $\KK_\eta$ up to core equivalence. First, we introduce the following terminology. 

\begin{definition} \rm \label{D:structure_net}
The {\em structural subnetwork} $\sS_\GG$ of $\GG$ is the subnetwork whose nodes are either simple or in a cycle appendage path component $\BB_i$ (see Remark~\ref{R:AiBi}) and whose arrows are arrows in $\GG$ that connect nodes in $\sS_\GG$.
\end{definition}

Lemma \ref{L:H'block} implies that all structural homeostasis subnetworks are contained in $\sS_\GG$, which is an input-output network. That is,  $\GG$ and $\sS_\GG$ have the same simple, super-simple, input, and output  nodes.  Lemma~\ref{L:supersimpleb} shows that every non-super-simple simple node lies between two adjacent super-simple nodes.  Using this fact, we can define a subnetwork $\LL$ of $\sS_\GG$ for every pair of adjacent super-simple nodes.

\begin{definition}\rm \label{D:supsimpnet}
Let $\rho_1, \rho_2$ be adjacent super-simple nodes.
\begin{enumerate}[label=(\alph*)]
\item A simple node $\rho$ is {\em between} $\rho_1$ and $\rho_2$ if there exists an $\iota o$-simple path that includes $\rho_1$ to $\rho$  to $\rho_2$ in that order. 
\item The {\em super-simple subnetwork}, denoted $\LL(\rho_1,\rho_2)$, is the subnetwork whose nodes are simple nodes between $\rho_1$ and $\rho_2$ and whose arrows are arrows of $\GG$ connecting nodes in $\LL(\rho_1,\rho_2)$.
\end{enumerate}
\end{definition}

It follows that all $\LL(\rho_1,\rho_2)$ are contained in $\sS_\GG$. By Lemma \ref{L:structural} (d), each appendage node in $\sS_\GG$ connects to exactly one $\LL$. This lets us expand a super-simple subnetwork $\LL\subset \sS_\GG$ to a super-simple structural subnetwork $\LL'\subset \sS_\GG$ as follows.

\begin{definition} \rm \label{def:L'}
Let $\rho_1$ and $\rho_2$ be adjacent super-simple nodes in $\GG$.  The \textit{super-simple structural subnetwork} $\LL'(\rho_1,\rho_2)$ is the input-output subnetwork consisting of nodes in $\LL(\rho_1,\rho_2)\cup \BB$ where $\BB$ consists of all appendage nodes that form cycles with nodes in $\LL(\rho_1,\rho_2)$; that is, all cycle appendage path components that connect to $\LL(\rho_1,\rho_2)$.   Arrows of $\LL'(\rho_1,\rho_2)$ are arrows of $\GG$ that connect nodes in $\LL'(\rho_1,\rho_2)$.  Note that $\rho_1$ is the input node and $\rho_2$ is the output node of $\LL'(\rho_1,\rho_2)$. 
\end{definition}

In Section~\ref{S:combinatorial_blocks} we prove that every subnetwork $\KK_\eta$ of $\GG$ associated with an irreducible structural homeostasis block $B_\eta$ is a super-simple structural subnetwork (Theorem~\ref{L:KetaL'}).  The converse is proved in Theorem~\ref{T:structural_supersimple}.

\subsection{Algorithm for enumerating homeostasis subnetworks} 
\label{sec:algorithm}

Before finding homeostasis in a model (say a biochemical model) one must choose input and output nodes (a modeling assumption) and reduce the resulting input-output network to a core network. Then we apply the following algorithm.

\paragraph{Step 0:} We begin by identifying the $\iota o$-simple paths in the core network and thus identifying the simple, super-simple, and appendage nodes. We also identify the appendage subnetwork $\AA_\GG$.

\paragraph{Step 1:} Determining the appendage homeostasis subnetworks from $\AA_\GG$.  Let 
\begin{equation} \label{e:AHB}
\AA_1\;,\ldots,\;\AA_s
\end{equation}
be the no cycle appendage path components of $\AA_\GG$ (see Remark~\ref{R:AiBi}). Theorem \ref{T:appendage_blocks} implies that these appendage path components are the subnetworks $\KK_\eta$ that correspond to appendage homeostasis blocks. In addition, there are $s$ independent defining conditions for appendage homeostasis given by the determinants of the Jacobian matrices $\det(J_{\AA_i}) = 0$ for $i = 1,\ldots, s$.

\paragraph{Step 2:} Determining the structural homeostasis subnetworks from $\sS_\GG$ (see Definition \ref{D:structure_net}). 
Let $\iota = \rho_1 > \rho_2 > \ldots > \rho_{q+1} = o$ be the super-simple nodes in $\sS_\GG$ in downstream order. Theorems \ref{L:KetaL'} and \ref{T:structural_supersimple} imply that up to core equivalence the $q$ super-simple structural subnetworks 
\begin{equation} \label{e:SHB}
\LL'(\iota,\rho_2),\; \LL'(\rho_2,\rho_3) \;,\ldots,\; \LL'(\rho_{q-1},\rho_q),\; \LL'(\rho_q,o)
\end{equation}
are the subnetworks $\KK_\eta$ that correspond to structural homeostasis blocks. In addition, there are  $q$ defining conditions for structural homeostasis blocks given by the determinants of the homeostasis matrices of the input-output networks: $\det\big(H(\LL'(\rho_i,\rho_{i+1}))\big) = 0$ for $i = 1,\ldots, q$.

Therefore, the $m = s + q$ subnetworks listed in \eqref{e:AHB} and \eqref{e:SHB} enumerate the appendage and structural homeostasis subnetworks in $\GG$.

\subsection{Low degree homeostasis types}
\label{S:LDHT}

Here we specialize our discussion to the low degree cases $k=1$ and $k=2$ where we determine all such homeostasis types (see Figure~\ref{F:3node}).  The first three types appear in three node networks and are given in the classification in \cite{GW19}. The fourth type has degree $k = 2$, but can only appear in networks with at least four nodes (see Figure~\ref{E:2A}).  We note that the lowest degree of a structural homeostasis block with an appendage node (that is, $\LL'\supsetneq \LL$) is $k = 3$ (see Figure~\ref{E:SWA}).

\begin{figure}[!ht]
\centering
\begin{subfigure}{1\textwidth}
\centering
\includegraphics[width=0.4\textwidth]{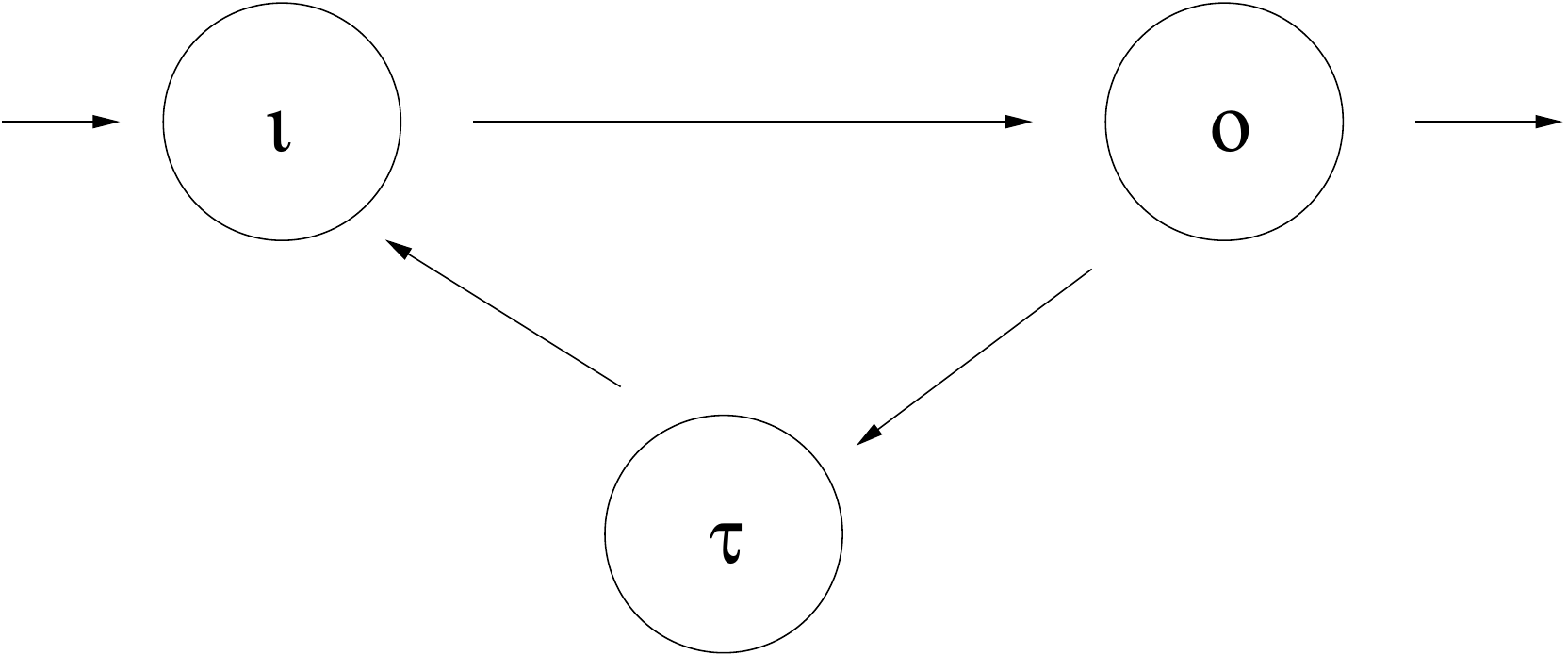}
\caption{\centering negative feedback loop \label{E:ND}}
\end{subfigure} \\[2mm]
\begin{subfigure}{0.4\textwidth}
\centering
\includegraphics[width=1.0\textwidth]{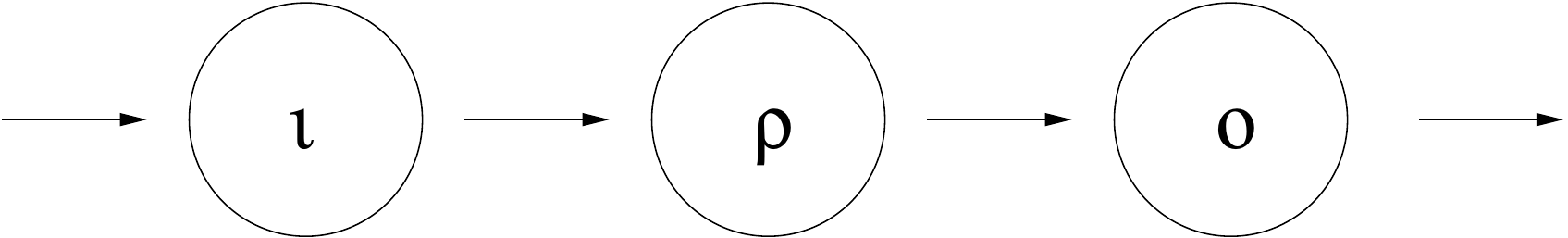}
\caption{kinetic motif \label{E:haldane}}
\end{subfigure} \qquad
\begin{subfigure}{0.4\textwidth}
\centering
\includegraphics[width=1.0\textwidth]{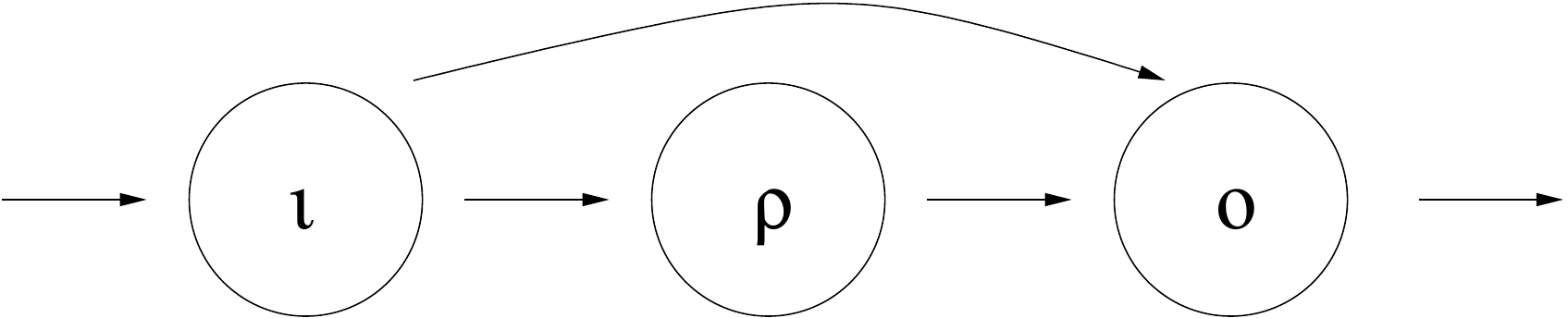}
\caption{(incoherent) feedforward loop \\[0mm] \rule{0mm}{0mm} \label{E:FFL}}
\end{subfigure}
\caption{{\bf Homeostasis types in three-node networks.}
(a) three-node core network exhibiting Haldane ($\iota\to o$) and null-degradation ($\tau$) homeostasis.
(b) three-node core network exhibiting Haldane $(\iota\to\rho$; $\rho\to o)$ homeostasis.
(c) three-node core network exhibiting degree 2 structural homeostasis.
According to \cite{GW19} this is a list of all three-node core networks up to core equivalence. \label{F:3node}}
\end{figure}

\subsubsection*{Degree 1 no cycle appendage homeostasis {\em (null-degradation)}} 
This corresponds to the vanishing of a degree $1$ irreducible factor of the form $(f_{\tau,x_\tau})$.
The single node $\tau$ is a no cycle appendage path component.  Apply Step 1 in the 
algorithm in Sections~\ref{SS:topological}-\ref{sec:algorithm} to Figure~\ref{E:ND}.

\subsubsection*{Degree 1 structural homeostasis {\em (Haldane)}} 
This corresponds to the vanishing of a degree $1$ irreducible factor of the form $(f_{j, x_\ell})$ whose associated subnetwork is $\LL'(\ell,j)$ of the form $\ell\to j$.  Apply Step 2 in the algorithm in Sections~\ref{SS:topological}-\ref{sec:algorithm} to Figure~\ref{E:haldane}.

\subsubsection*{Degree 2 structural homeostasis {\em (feedforward loop)}}
This corresponds to a three-node input-output subnetwork $\LL'(\ell,j)$ with input node $\ell$, output node $j$, and regulatory node $\rho$, where $\ell$ and $j$ are adjacent super-simple and $\rho$ is a simple node between the two super-simple nodes.  It follows that both paths $\ell\to\rho\to j$ and $\ell\to j$ are in $\LL' = \LL$. Hence, $\LL'$ is a feedforward loop motif. Homeostasis occurs when
\[
\det\big(H(\LL'(\ell,j))\big) = f_{\rho,x_{\ell}} f_{j,x_{\rho}} - f_{j,x_{\ell}}f_{\rho,x_{\rho}} = 0
\]
Apply Step 2 in the algorithm in Sections \ref{SS:topological}-\ref{sec:algorithm} to Figure~\ref{E:FFL}.

\subsubsection*{Degree 2 no cycle appendage homeostasis} 
This is associated with a two-node appendage path component 
$\AA = \{\tau_1, \tau_2\}$ with arrows $\tau_1 \to \tau_2$ and $\tau_2 \to \tau_1$. Homeostasis occurs when 
\[
\det\big(J(\AA)\big) = f_{\tau_1, x_{\tau_1}} f_{\tau_2, x_{\tau_2}} - f_{\tau_1, x_{\tau_2}}  f_{\tau_2, x_{\tau_1}} = 0
\]
Apply Step 1 in the algorithm in Sections~\ref{SS:topological}-\ref{sec:algorithm} to Figure~\ref{E:2A}.

\begin{figure}[!ht]
\centering
\begin{subfigure}{0.4\textwidth}
\centering
\includegraphics[width=1.0\textwidth]{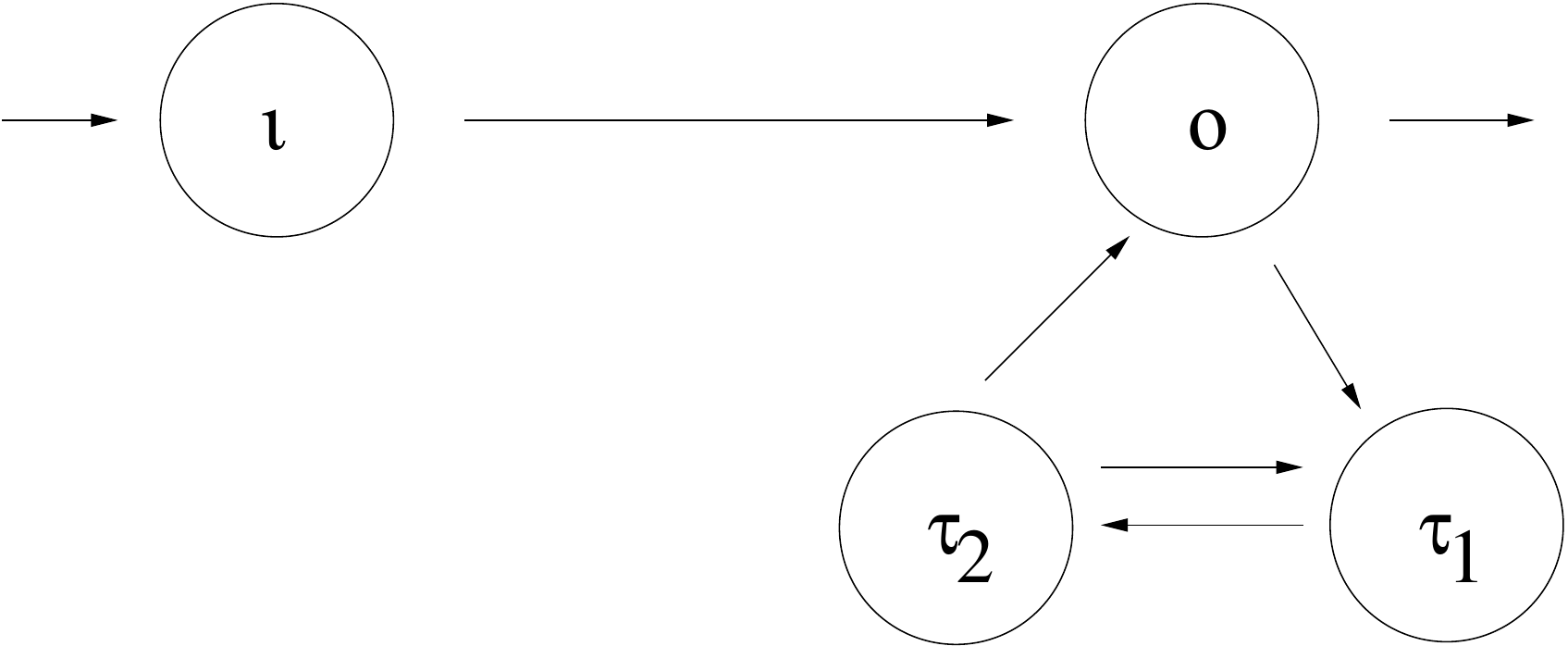}
\caption{Haldane ($\iota\to o$); Degree 2 no cycle appendage ($\tau_2\Leftrightarrow\tau_1$)  \label{E:2A}}
\end{subfigure} \qquad
\begin{subfigure}{0.4\textwidth}
\centering
\includegraphics[width=1.0\textwidth]{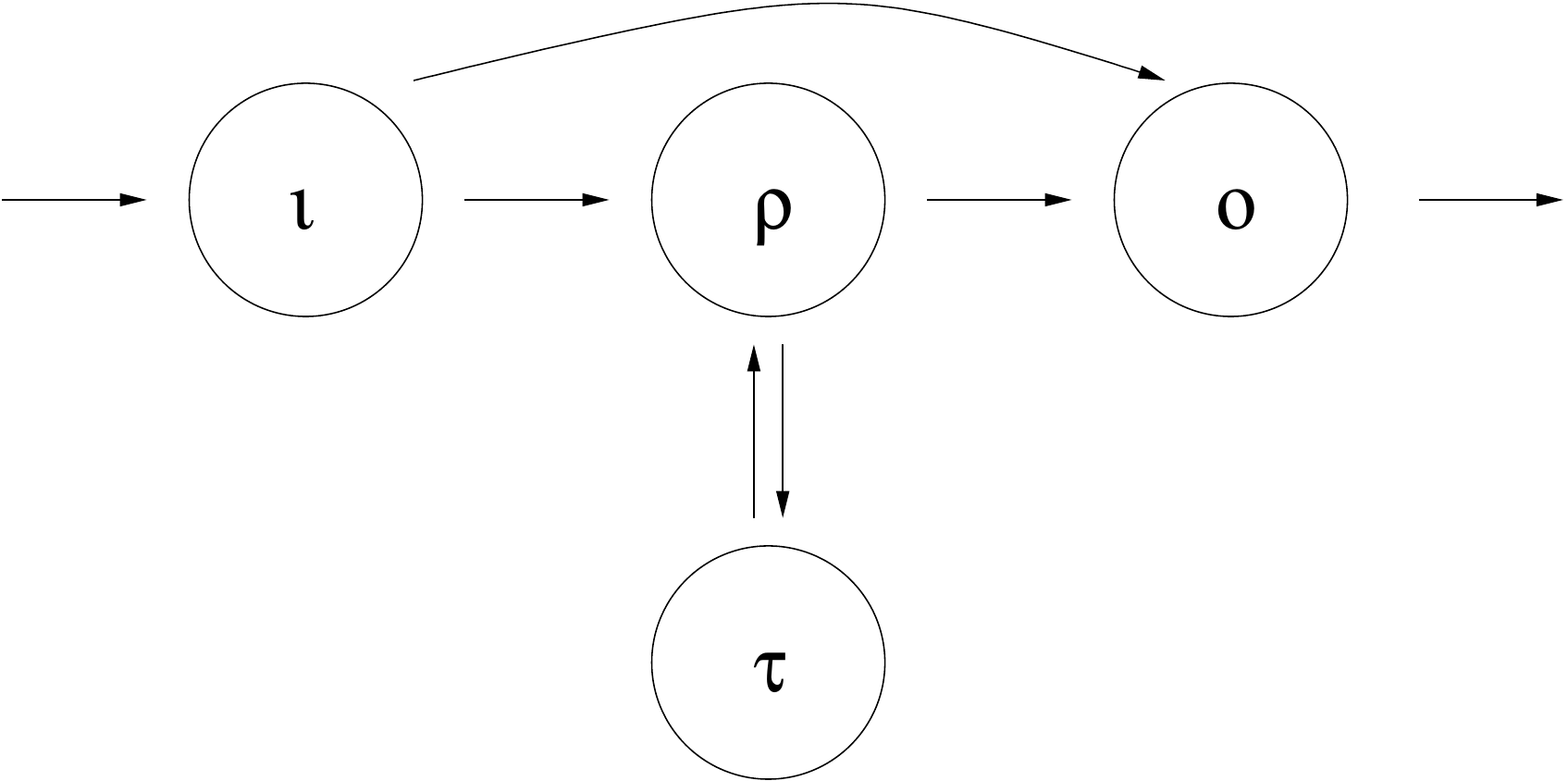}
\caption{Structural with appendage node \\[1mm] \rule{0cm}{0mm} \label{E:SWA}}
\end{subfigure}
\caption{{\bf Homeostasis types in four-node networks.}
	(a) smallest network exhibiting degree 2 no cycle appendage homeostasis. 
	(b) smallest network exhibiting appendage node in structural homeostasis. 
        \label{F:4node}}
\end{figure}

\subsection{12-node artificial network example}
\label{sec:12-node}

We now return to the artificial example in Subsection~\ref{SSS:example} to illustrate the algorithm for enumerating homeostasis blocks.  The network shown in Figure~\ref{F:homeostasis-example-comp} has input node ($\iota$), output node ($o$), six simple nodes ($\rho_1,\ldots,\rho_6$), and four appendage nodes ($\tau_1,\tau_2,\tau_3,\tau_4$).   The input-output network $\GG$ in Figure~\ref{F:homeostasis-example-comp} has four $\iota o$-simple paths (see Table~\ref{T:iota_o_simple}) and six homeostasis subnetworks that can be found in two steps using the algorithm in Section~\ref{sec:algorithm}.

\paragraph{Step 1:} $\GG$ has three appendage path components ($\AA_1 = \{\tau_1\}$, $\AA_2 = \{\tau_2,\tau_3\}$, $\BB_1 = \{\tau_4\}$) in $\AA_\GG$. Among these, $\AA_1$ and $\AA_2$ satisfy the no cycle condition, whereas $\BB_1$ does not since $\tau_4$ forms a cycle with simple node $\rho_6$.  Hence, there are two appendage homeostasis subnetworks given by $\AA_1$ and $\AA_2$. 

\paragraph{Step 2:} $\GG$ has five super-simple nodes (in downstream order, they are $\iota, \rho_1, \rho_3, \rho_4, o$). The five super-simple nodes lead to four structural homeostasis subnetworks given (up to core equivalence) by $\LL'(\iota,\rho_1)$, $\LL'(\rho_1,\rho_3)$,  $\LL'(\rho_3,\rho_4)$, $\LL'(\rho_4,o)$.

Table \ref{T:HT12} lists the six homeostasis subnetworks in $\GG$, which give the six irreducible factors of $\det(H)$ where $H$ is the $11\times 11$ homeostasis matrix of $\GG$. The factorization of the degree $11$ homogeneous polynomial $\det(H)$ is given by
\[
\det(\PH)  = \pm  f_{\tau_1,x_{\tau_1}}
\det(B_2) f_{\rho_1,x_{\iota}} 
\det(B_4) f_{\rho_4,x_{\rho_3}} 
\det(B_6) 
\]
where 
\begin{equation*} 
\begin{split}
B_2 = & \Matrixc{ f_{\tau_2,x_{\tau_2}}  & f_{\tau_2,x_{\tau_3}} \\ f_{\tau_3,x_{\tau_2}}  & f_{\tau_3,x_{\tau_3}}}
\qquad
B_4 = \Matrixc{ f_{\rho_2,x_{\rho_1}}  & f_{\rho_2,x_{\rho_2}} \\ f_{\rho_3,x_{\rho_1}}  & f_{\rho_3,x_{\rho_2}}} \\[2mm]
& B_6 = \Matrixc{
	f_{\rho_5,x_{\rho_4}} & f_{\rho_5,x_{\rho_5}} & 0 & 0 \\
	f_{\rho_6,x_{\rho_4}}  & 0 &f_{\rho_6,x_{\rho_6}}  &f_{\rho_6,x_{\tau_4}}  \\
	0 & 0 & f_{\tau_4,x_{\rho_6}} & f_{\tau_4,x_{\tau_4}} \\
	0 & f_{o,x_{\rho_5}} & f_{o,x_{\rho_6}} & 0 
}
\end{split}      
\end{equation*}

\begin{table}[!ht]
\begin{center}
\caption{{\bf Homeostasis subnetworks in Figure~\ref{F:homeostasis-example-comp}.} \label{T:HT12}}
\begin{tabular}{|c|c|l|}
\hline
Class & Homeostasis subnetworks & Name \\ \thickhline
appendage & $\AA_1 = \{\tau_1\}$ &  null-degradation\\ \hline
appendage & $\AA_2 = \{\tau_1\Leftrightarrow \tau_2\}$ & no cycle appendage \\ \hline 
structural & $\LL'(\iota,\rho_1) = \{\iota\to\rho_1\}$& Haldane \\ \hline
structural & $\LL'(\rho_1,\rho_3) = \{\rho_1, \rho_2\, \rho_3 \}$  & feedforward loop \\ \hline
structural & $\LL'(\rho_3,\rho_4) = \{\rho_3\to \rho_4\}$ & Haldane \\ \hline
structural  & $\LL'(\rho_4,o) = \{\rho_4, \rho_5, \rho_6, \tau_4, o\}$ & degree 4 structural \\ \hline
\end{tabular}
\end{center}
\end{table}

\subsection{A biological example}
\label{sec:bionetwork}

The first step in applying our algorithm for finding infinitesimal homeostasis to a biological input-output network is to convert the network to a mathematical input-output network and then, if necessary, reducing the math network to a core network. See \cite{GW19} for the application of our methods to three-node biochemical systems. As another application, we consider the five-node {\em E. coli} chemotaxis network studied in Ma et al. \cite{MTELT09} (see Figure~\ref{F:biochemical-example}, left panel). In this example the input node is the {\em Receptor complex} and the output node is the response regulator {\em CheY}. 

The corresponding $5$-node mathematical network is shown in the middle panel of Figure~\ref{F:biochemical-example}. This network can be reduced to a $4$-node core network (Figure~\ref{F:biochemical-example}, right panel)  by removing the node $\tau_3$, which is not downstream from the input node, and the arrow $\tau_3 \to \tau_2$. The remaining nodes are both downstream from $\iota$ and upstream from $o$ and hence form a core network $\GG_c$ (see Definition \ref{D:core}). 

The core network $\GG_c$ has one $\iota o$-simple path $\iota\to o$ with $\iota$ and $o$ being the super-simple nodes. The appendage subnetwork $\AA_{\GG_c}$ consists of two appendage nodes $\tau_1$ and $\tau_2$. We enumerate the homeostasis blocks in two steps:
\begin{enumerate}
	\item $\AA_{\GG_c}$ has two appendage path components ($\AA_1 = \{ \tau_1\}$, $\AA_2=\{\tau_2\}$) and each component satisfies the no-cycle condition. Hence, there are two appendage homeostasis subnetworks given by $\AA_1$ and $\AA_2$.
	\item The two super-simple nodes lead to only one structural homeostasis subnetwork given by $\LL(\iota, o)=\LL'(\iota, o)$, which is the $\iota o$-simple path. 
\end{enumerate} 
These three homeostasis subnetworks give degree 1 factors of $\det(H)$ where $H$ is the $3\times 3$ homeostasis matrix of $\GG_c$. It follows that two types of homeostasis can occur in this {\em E. coli} network: null-degradation homeostasis occurs when $f_{\tau_1,x_{\tau_1}}$ or $f_{\tau_2,x_{\tau_2}}$ vanish (that is, the linearized
internal dynamics of {\em Methylation level} or {\em CheB} is zero) and Haldane homeostasis occurs when $f_{o,x_\iota} = 0$ (that is, the coupling from the input node {\em Receptor complex} to the output node {\em CheY} is $0$).

	\begin{figure}[!ht]
\centering
	\includegraphics[width=.33\linewidth]{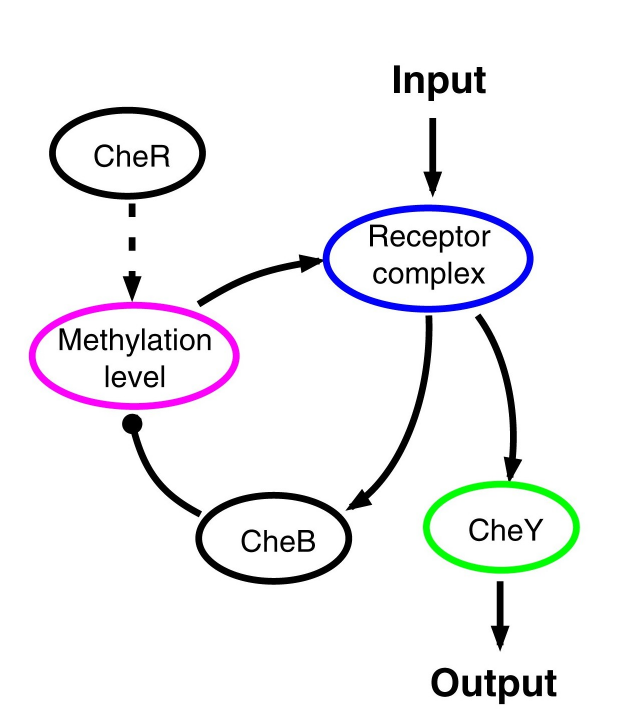}\hspace{.2in}
	\includegraphics[width=.26\linewidth]{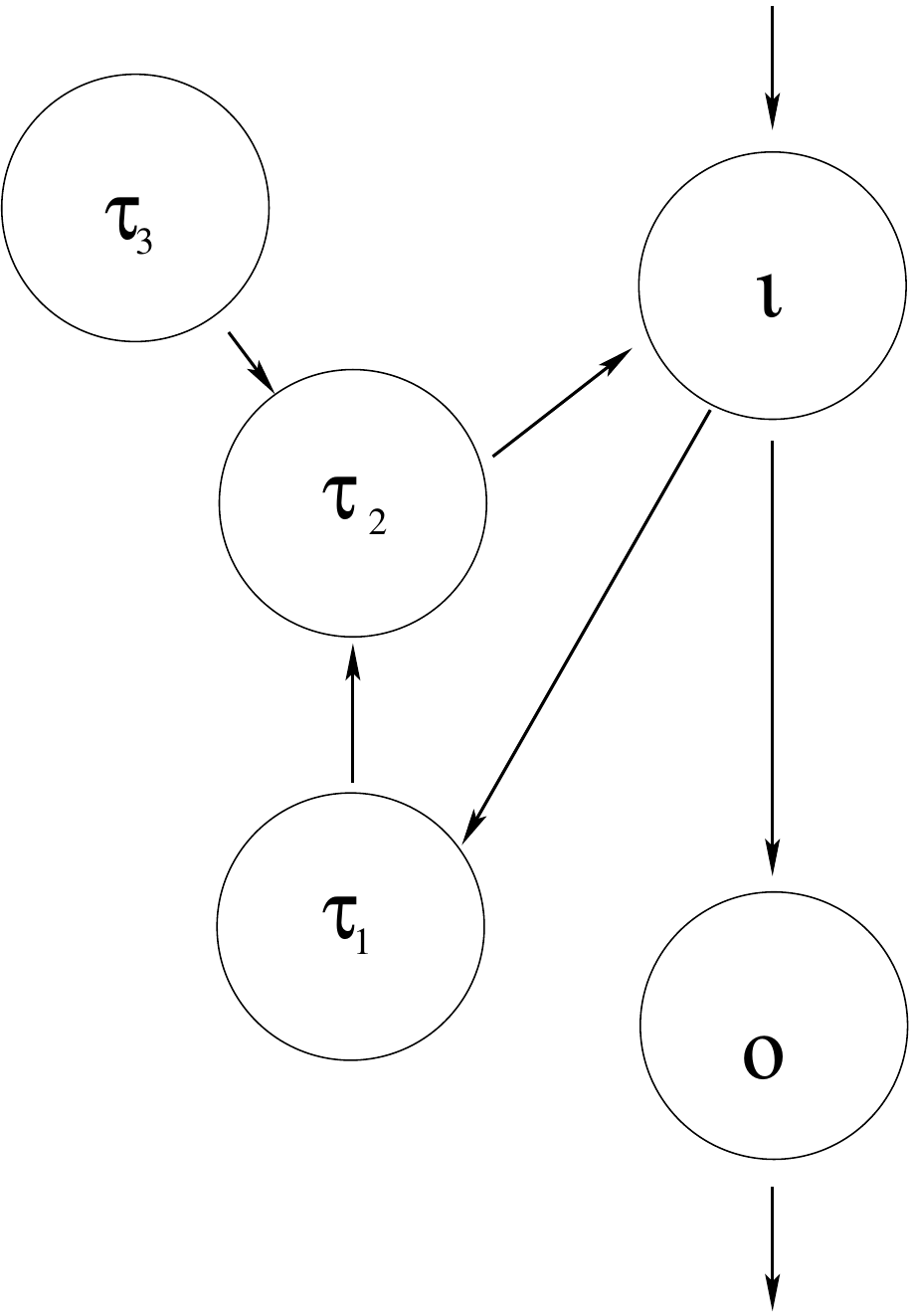}\hspace{.4in}
	\includegraphics[width=.2\linewidth]{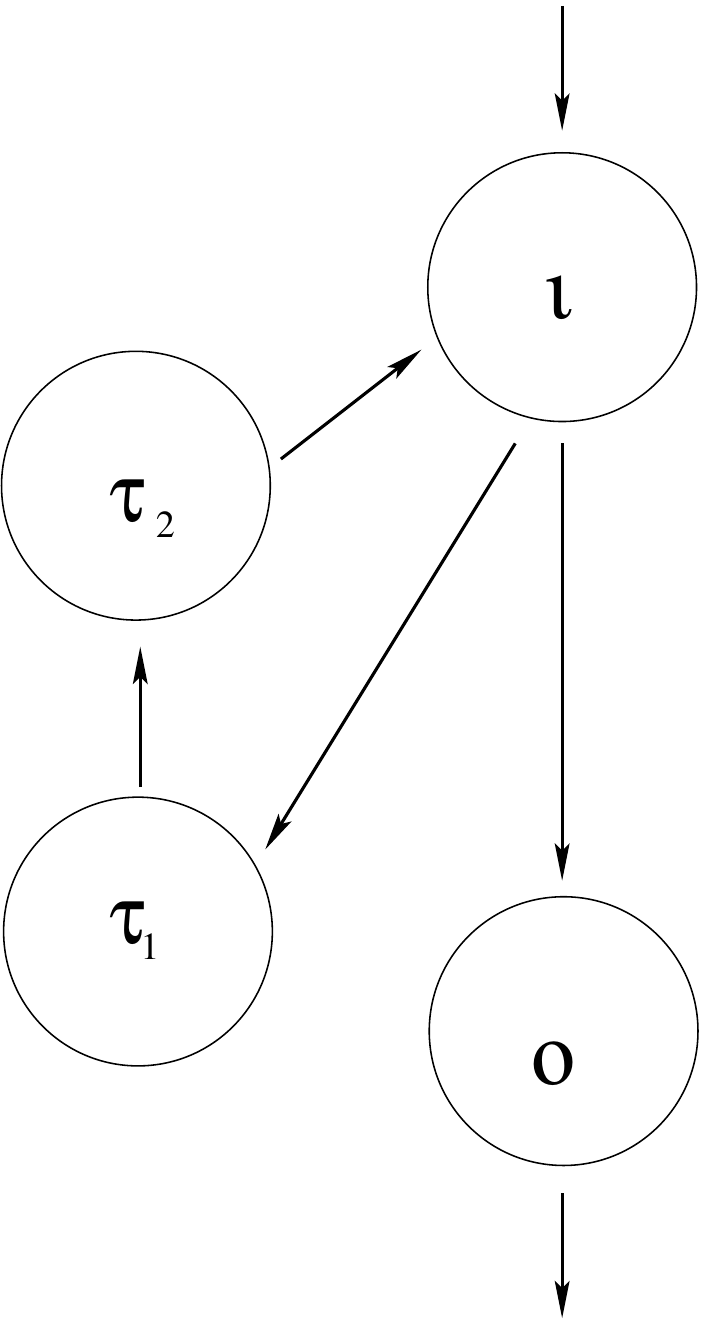}
		\caption{{\bf A biological network example.} (Left) The {\em E.~coli} chemotaxis network from Ma et al. \cite{MTELT09}. (Middle) The mathematical input-output network $\GG$ corresponding to the {\em E. coli} network. (Right) The core subnetwork $\GG_c$ of $\GG$.  \label{F:biochemical-example}} 
	\end{figure}

The analysis in Ma et al.~\cite{MTELT09} proceeds along a slightly different tack.  There the authors simplify the network to $3$-nodes by combining $\tau_1$ and $\tau_2$.  This changes the outcome whereby null-degradation can only occur in one way in their formulation and it is through the simultaneous occurrence of null-degradation in $\tau_1$ and $\tau_2$.

\subsection{Remark on chairs}

Nijhout, Best and Reed~\cite{NBR14} observed that homeostasis often appears in models in the form of a {\em chair}.  That is, as $\II$ varies, the input-output function $x_0(\II)$ has the piecewise linear description: increases linearly, is approximately constant, and then increases linearly again. Golubitsky and Stewart~\cite{GS17} observed that it follows from elementary catastrophe theory that smooth {\em chair singularities} have the normal form $\II^3$, defining conditions 
\[
x_o'(\II_0) = x_o''(\II_0) = 0
\]
and nondegeneracy condition $x_0'''(\II_0) \neq 0$.  Moreover, \cite{GW19} noted that if $x_0'(\II) = g(\II) h(\II)$, where $g(\II_0) \neq 0$, then the defining conditions for a chair singularity are equivalent to 
\begin{equation} \label{eq:chair_conditions}
h(\II_0) = h'(\II_0) = 0 \qquad\text{and}\qquad h''(\II_0) \neq 0
\end{equation}
It follows from Lemma~\ref{L:det} and \eqref{eq:FK_factors} that a chair singularity for infinitesimal homeostasis is of type $B_\eta$ if $h_\eta(\II)$ satisfies~\eqref{eq:chair_conditions} at $\II=\II_0$.

\subsection{Remarks on the interpretation of structural and appendage homeostasis}
\label{SS:FFFB}

We claim that structural homeostasis balances fluxes along simple paths from one super-simple node to the next.  This is defined by $\det(H(\LL'(\rho_j,\rho_{j+1})))$ and gives a feedforward interpretation to structural homeostasis.  See Definition~\ref{def:L'}. The balancing can be weighted by appendage nodes that appear in $\LL'(\rho_j,\rho_{j+1})\setminus\LL(\rho_j,\rho_{j+1})$.  

On the other hand, appendage homeostasis balances fluxes in a given appendage path component. Each of these  path components has input from a super-simple node and output to an upstream super-simple node, and this gives a feedback interpretation.

\subsection{Structure of the paper}
\label{SS:C}

In Section~\ref{sec:core} we show that infinitesimal homeostasis in the original system \eqref{eq:ad_coord} occurs in a network if and only if infinitesimal homeostasis occurs in the core network for the associated frozen system.  See Theorem~\ref{T:core}.  We discuss when backward arrows can be ignored when computing the determinant of the homeostasis matrix and the limitations of this procedure. See Corollary~\ref{P:core_equivalent}.  In Section~\ref{S:determinant} we relate the form of the summands of the determinant of the homeostasis matrix $H$ with the form of $\iota o$-simple paths of the input-output network.  See Theorem~\ref{L:summand_form}.  In Theorem~\ref{T:core_equivalent} we also discuss `core equivalence'.  In Section~\ref{S:classification2} we prove the theorems about the appendage and structural classes of homeostasis.  See Definition~\ref{D:homeostasis_types}, Theorem~\ref{L:self_couplings}, and the normal form Theorem~\ref{lem:associated_network}.  In Section~\ref{S:appendage_blocks} we prove the necessary conditions that appendage homeostasis must satisfy.  See Theorem~\ref{thm:appendage}.  In Section~\ref{S:combinatorial_blocks}, specifically Section~\ref{S:structural_blocks}, we introduce an ordering of super-simple nodes that leads to a combinatorial definition of structural blocks.  See Definition~\ref{D:supsimpnet} and Definition~\ref{def:L'}.  The connection of these blocks with the subnetworks $\KK_\eta$ obtained from the homeostasis matrix is given in Corollary~\ref{cor:H'} and Theorem~\ref{L:KetaL'}.  In Section~\ref{S:CC} we summarize our algorithm for finding infinitesimal homeostasis directly from the input-output network $\GG$.  It also gives a topological classification of the different types of infinitesimal homeostasis that the network $\GG$ can support.

\Section{Core networks}
\label{sec:core}

Let $\GG$ be an input-output network with input node $\iota$, output node $o$, and regulatory nodes $\rho_j$.
We use the notions of upstream and downstream nodes to construct a core subnetwork $\GG_c$ of $\GG$.

The stable equilibrium $(X_0,\II_0)$ of the system of differential equations \eqref{eq:ad_coord} satisfy a system of nonlinear equations~\eqref{eq:F=0}, that can be explicitly written as 
\begin{equation} \label{eq:ad_coord0}
\begin{array}{lcl}
f_\iota(x_\iota,x_\rho, x_o,\II) & = & 0\\
f_\rho(x_\iota,x_\rho, x_o) & = & 0\\
f_\iota(x_\iota,x_\rho, x_o) & = & 0
\end{array}
\end{equation}
We start by partitioning the regulatory nodes $\rho$ into three types: 
\begin{itemize}
\item those nodes $\sigma$ that are both upstream from $o$ and downstream from $\iota$, 
\item those nodes $d$ that are not downstream from $\iota$,  
\item those nodes $u$ that are downstream from $\iota$ and not upstream from $o$.
\end{itemize}
Based on this partition, the system \eqref{eq:ad_coord0} has the form 
\begin{equation} \label{eq:ad_coord2}
\begin{array}{lcl}
f_\iota(x_\iota,x_\sigma,x_u, x_d, x_o,\II) & = & 0\\
f_\sigma(x_\iota,x_\sigma,x_u, x_d, x_o) & = & 0 \\
f_u(x_\iota,x_\sigma,x_u, x_d, x_o) & = & 0 \\
f_d(x_\iota,x_\sigma,x_u, x_d, x_o) & = & 0 \\
f_o(x_\iota,x_\sigma,x_u, x_d, x_o)  & = & 0
\end{array}
\end{equation}
In Lemma~\ref{L:missing_arrows} we make this form more explicit.

\begin{lemma} \label{L:missing_arrows}
The definitions of $\sigma$, $u$, and $d$ nodes imply the admissible system \eqref{eq:ad_coord2} has the form
\begin{equation} \label{eq:ad_coord2A}
\begin{array}{rcl}
\dot{x}_\iota & = & f_\iota(x_\iota,x_\sigma, x_d, x_o,\II) \\
\dot{x}_\sigma & = & f_\sigma(x_\iota,x_\sigma, x_d, x_o) \\
\dot{x}_u & = & f_u(x_\iota,x_\sigma,x_u, x_d, x_o) \\
\dot{x}_d & = & f_d(x_d) \\
\dot{x}_o & = & f_o(x_\iota,x_\sigma, x_d, x_o)
\end{array}
\end{equation}
Specifically, arrows of type $\sigma\to d$, $\iota\to d$, $u\to d$, $o\to d$, $u\to\sigma$,  $u\to o$, $u\to\iota$ do not exist.
\end{lemma}
\begin{proof} 
We list the restrictions on \eqref{eq:ad_coord2} given first by the definition of $d$ and then by the definition of $u$.
	
\begin{description} 
\item[$\sigma\not\to d$]  If a node in $\sigma$ connects to a node in $d$, then there would be a path from $\iota$ to a node in $d$ and that node in $d$ would be downstream from $\iota$, a contradiction.  Therefore, $f_d$ is independent of $x_\sigma$.
\item[$\iota\not\to d$] Similarly, the node $\iota$ cannot connect to a node in $d$, because that node would then be downtream from $\iota$, a contradiction.  Therefore, $f_d$ is independent of $x_\iota$.  
\item[$o\not\to d$] If there is an arrow $o\to d$, then there is a path $\iota\to\sigma\to o\to d$.  Hence there is a path $\iota \to d$ and that is not allowed.  Therefore, $f_d$ is independent of $x_o$.  
\item[$u\not\to d$] Note that nodes in $u$ must be downstream from $\iota$. Hence, there cannot be a connection from $u$ to $d$ or else there would be a connection from $\iota$ to $d$.  Therefore, $f_d$ is independent of $x_u$.
\item[$u\not\to\sigma$] if a node in $u$ connects to a node in $\sigma$, then there would be a path from $u$ to $o$ and $u$ would be upstream from $o$, a contradiction.  Therefore, $f_\sigma$ is independent of $x_u$.  
\item[$u\not\to o$] Suppose a node in $u$ connects to $o$. Then that node is upstream from $o$, a contradiction.  Therefore, $f_o$ is independent of $x_u$.  
\item[$u\not\to\iota$] Finally, if $u$ connects to $\iota$, then $u$ connects to $o$, a contradiction.  Therefore, $f_\iota$ is independent of $x_u$.  
\end{description}
The remaining types of connections can exist in $\GG_c$.
Nodes and arrows that can exist in $\GG_c$ are shown in Figure~\ref{F:isudoGC}. 
\qed
\end{proof}

\begin{figure}[!ht]
\centering
\begin{subfigure}{0.4\textwidth}
\includegraphics[width=1\textwidth]{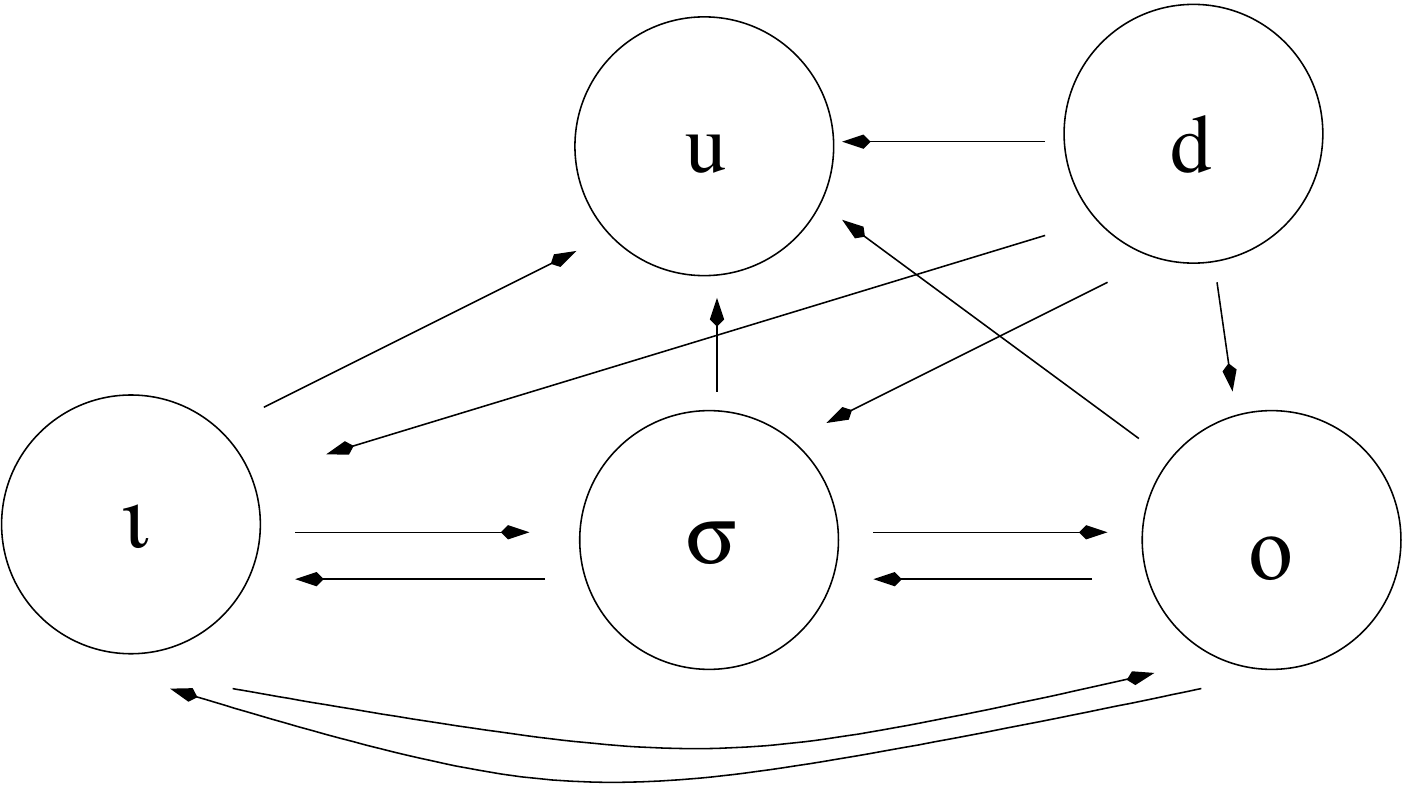}
\caption{Partision of input-output network into node types. \label{F:isudoG}}
\end{subfigure}\qquad 
\begin{subfigure}{0.5\textwidth}
\includegraphics[width=1\textwidth]{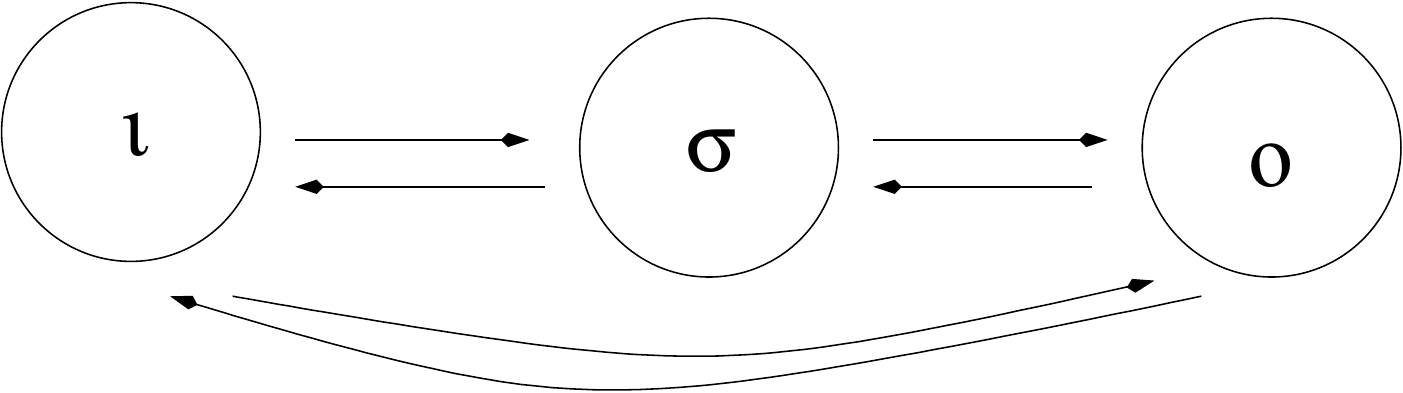}
\caption{Partition of core network into node types. \label{F:isudoGC}}
\end{subfigure}
\caption{Nodes and arrows in general network and core network.  \label{F:isudo}} 
\end{figure}

\begin{lemma} \label{core1jacobian}
Suppose $X_0=(x^*_\iota, x^*_\sigma, x^*_u, x^*_d, x^*_o)$ is a stable equilibrium of \eqref{eq:ad_coord2A}.   Then the core admissible system (obtained by freezing the $x_d$ nodes at $x^*_d$ and deleting the $x_u$ nodes)
\begin{equation} \label{eq:ad_coord3}
\begin{array}{rcl}
\dot{x}_\iota & = & f_\iota(x_\iota,x_\sigma,x^*_d, x_o,\II) \\
\dot{x}_\sigma & = & f_\sigma(x_\iota,x_\sigma,x^*_d, x_o) \\
\dot{x}_o & = & f_o(x_\iota,x_\sigma,x^*_d, x_o) 
\end{array}
\end{equation}
has a stable equilibrium at  $Y_0=(x^*_\iota,x^*_\sigma,x^*_o)$.  
\end{lemma}
\begin{proof}
It is straightforward that $Y_0$ is an equilibrium of \eqref{eq:ad_coord3}.  Reorder coordinates $(\iota, \sigma, u, d, o)$ to $(\iota,\sigma,o,d,u)$.  Then Lemma~\ref{L:missing_arrows} implies that the Jacobian $J$ of \eqref{eq:ad_coord2A} has the form
\begin{equation} \label{e:J3}
J = \Matrixc{f_{\iota, x_\iota} &  f_{\iota, x_\sigma} &  0 & f_{\iota, x_d} & f_{\iota, x_o} \\ 
		f_{\sigma, x_\iota} &  f_{\sigma, x_\sigma} &  0 & f_{\sigma, x_d} & f_{\sigma, x_o} \\
		f_{u, x_\iota} &  f_{u, x_\sigma} &  f_{u, x_u}  &  f_{u, x_d}  & f_{u, x_o} \\
		0 &  0 &  0 & f_{d, x_d}  & 0 \\
		f_{o, x_\iota} &  f_{o, x_\sigma} & 0 &  f_{o, x_d} & f_{o, x_o}
}
\end{equation}
and on swapping the $u$ and $o$ coordinates we see that $J$ is similar to 
\begin{equation} \label{e:J4}
J_1 = \Matrixc{ f_{\iota, x_\iota}  & f_{\iota, x_\sigma}  & f_{\iota, x_o}  & f_{\iota, x_d}   & 0 \\ 
		f_{\sigma, x_\iota} & f_{\sigma, x_\sigma} & f_{\sigma, x_o}  & f_{\sigma, x_d}  & 0 \\
		f_{o, x_\iota}     & f_{o, x_\sigma}      & f_{o, x_o}      & f_{o, x_d}      & 0 \\
		0                 & 0                  & 0              & f_{d, x_d}      & 0 \\
		f_{u, x_\iota}     &  f_{u, x_\sigma}     & f_{u, x_o}      &  f_{u, x_d}     & f_{u, x_u}
}
\end{equation}
It follows that the eigenvalues of $J$ at $X_0$ are the eigenvalues of $f_{d,x_d}$, $f_{u,x_u}$, and the eigenvalues of the Jacobian of \eqref{eq:ad_coord3} at $Y_0$.  Since the eigenvalues of  $J_1$ have negative real part, the equilibrium $Y_0$ is stable.
\qed
\end{proof}

\begin{lemma} \label{L:core->general}
Suppose that $\GG$ is an input-output network with core network $\GG_c$.  Suppose that the core admissible system  
\begin{equation} \label{eq:ad_coord4}
\begin{array}{rcl}
\dot{x}_\iota & = & f_\iota(x_\iota,x_\sigma, x_o,\II) \\
\dot{x}_\sigma & = & f_\sigma(x_\iota,x_\sigma, x_o) \\
\dot{x}_o & = & f_o(x_\iota,x_\sigma, x_o) 
\end{array}
\end{equation}
has a stable equilibrium at  $Y_0=(x^*_\iota,x^*_\sigma,x^*_o)$ and a point of infinitesimal homeostasis at $\II_0$.  Then the admissible system for the original network $\GG$ can be taken to be
\begin{equation} \label{eq:ad_coord5}
\begin{array}{rcl}
\dot{x}_\iota & = & f_\iota(x_\iota,x_\sigma, x_o,\II) \\
\dot{x}_\sigma & = & f_\sigma(x_\iota,x_\sigma, x_o) \\
\dot{x}_d & = & -x_d \\
\dot{x}_u & = & -x_u \\
\dot{x}_o & = & f_o(x_\iota,x_\sigma, x_o) 
\end{array}
\end{equation} 
has a stable equilibrium at $X_0 = (x^*_\iota, x^*_\sigma, 0,0, x^*_o)$ and infinitesimal homeostasis at $\II_0$.  
\end{lemma}

\begin{theorem} \label{T:core}
Let $x_o(\II)$ be the input-output function of the admissible system \eqref{eq:ad_coord2A} and let $x_o^c(\II)$ be the input-output function of the associated core admissible system \eqref{eq:ad_coord3}.  Then the input-output function $x_o^c(\II)$ associated with the core subnetwork has a point of infinitesimal homeostasis at $\II_0$ if and only if the input-output function $x_o(\II)$ associated with the original network has a point of infinitesimal homeostasis at $\II_0$.  More precisely, 
\begin{equation}
x_o'(\II) = k(\II)\,x_o^{c\;\prime}(\II)
\end{equation}
where $k(\II_0)\neq 0$. 
\end{theorem}

\begin{proof}
It follows from Lemma~\ref{L:det} that $x'_o(\II_0) = 0$ if and only if 
\[
\det\Matrixc{
  f_{\sigma, x_\iota} &  f_{\sigma, x_\sigma} &  0 & f_{\sigma, x_d} \\
  f_{u, x_\iota} &  f_{u, x_\sigma} &  f_{u, x_u}  &  f_{u, x_d}   \\
  0 &  0 &  0 & f_{d, x_d}   \\
  f_{o, x_\iota} &  f_{o, x_\sigma} & 0 &  f_{o, x_d} } = 0
\]
if and only if 
\[
\det(f_{u,x_u})  \det\Matrixc{
  f_{\sigma, x_\iota} &  f_{\sigma, x_\sigma} & f_{\sigma, x_d} \\
  0 &  0 &  f_{d, x_d}   \\
  f_{o, x_\iota} &  f_{o, x_\sigma} &  f_{o, x_d} } = 0
\]
if and only if 
\[
\det(f_{u,x_u}) \det(f_{d,x_d}) \det\Matrixc{
  f_{\sigma, x_\iota} &  f_{\sigma, x_\sigma}  \\
  f_{o, x_\iota} &  f_{o, x_\sigma} } = 0
\]
Both matrices $f_{u,x_u}$ and $f_{d,x_d}$ are triangular with negative diagonal entries and thus have nonzero determinants. It then follows from Lemma~\ref{L:det} that $x_o^{c\;\prime}(\II_0) = 0$ if and only if 
\begin{equation} \label{eq:core1J}
\det\Matrixc{
  f_{\sigma, x_\iota} &  f_{\sigma, x_\sigma}  \\
  f_{o, x_\iota} &  f_{o, x_\sigma} } = 0
\end{equation}
is satisfied.
\qed
\end{proof}

It follows from Theorem~\ref{T:core} and Lemma~\ref{L:core->general} that classifying infinitesimal  homeostasis for networks $\GG$ is identical to classifying infinitesimal homeostasis for the core subnetwork $\GG_c$.  Specifically, an admissible system with infinitesimal homeostasis for the core subnetwork yields, by Lemma~\ref{L:core->general}, an admissible system with infinitesimal homeostasis for the original network which in turn yields the original system for the core subnetwork with infinitesimal homeostasis by Theorem~\ref{T:core}.

\ignore{
\begin{remark} \rm \label{R:backward}
Corollary~\ref{P:core_equivalent} does not imply that backward arrows can be eliminated in order to define the core network.  The reason is that stability of the equilibrium on the network without the backward arrows might not be the same as stability on the original network.  The stability of the equilibrium is necessary to get a well defined input-output function.  See Example~\ref{E:backward}.
\end{remark}  
}

\begin{remark} \rm \label{R:backward}
Corollary~\ref{P:core_equivalent} implies that backward arrows can be eliminated when computing zeros of $\det(H)$. These arrows cannot be eliminated when computing equilibria of the network equations or their stability. See \eqref{e:stab_ex} in Example~\ref{E:backward}.
\end{remark}  

\begin{example} \rm \label{E:backward}
Consider the network in Figure~\ref{F:backward}.  Assume WLOG that an admissible vector field for this network
\begin{equation} \label{eq:ad_net_back}
\begin{array}{rcl}
\dot{x}_\iota & = & f_\iota(x_\iota,x_\rho,\II) \\
\dot{x}_\rho & = & f_\rho(x_\iota,x_\rho) \\
\dot{x}_o & = & f_o(x_\rho, x_o) 
\end{array}
\end{equation}
has an equilibrium at the origin $(X_0,\II_0)=(0,0)$; that is 
\[
f_\iota(0,0,0) = f_\rho(0,0) = f_o(0,0) = 0.
\]
Begin by noting that the Jacobian of \eqref{eq:ad_net_back} is
\begin{equation} \label{e:J}
J = \Matrixc{f_{\iota, x_\iota} &  f_{\iota, x_\rho} & 0  \\ 
f_{\rho, x_\iota} &  f_{\rho, x_\rho} & 0 \\
0 &  f_{o, x_\rho} & f_{o, x_o}}
\end{equation}

\begin{figure}[!ht]
\centering
\includegraphics[width=0.5\textwidth]{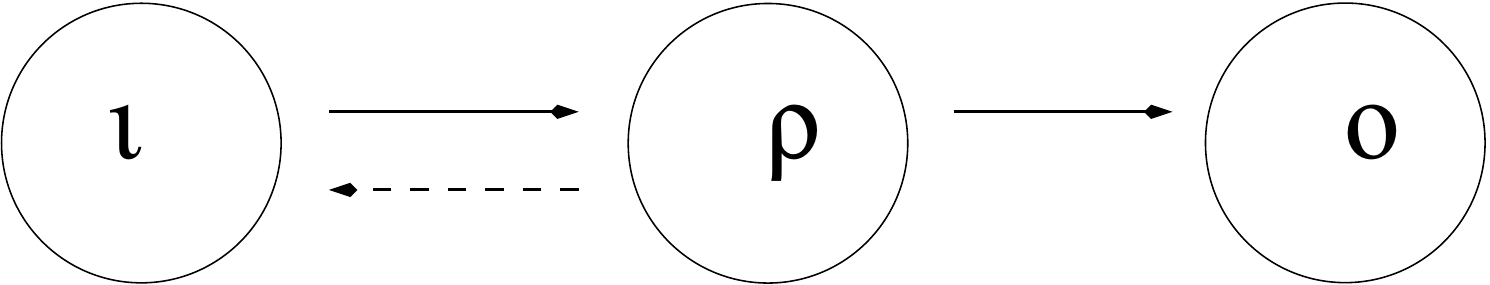}
\caption{{\bf Backward arrow.} Network with a (dashed) backward arrow. \label{F:backward}} 
\end{figure}
	
The origin is a linearly stable equilibrium if and only if 
\begin{equation} \label{e:stab_ex}
f_{o,x_0} < 0 \qquad f_{\iota, x_\iota}  + f_{\rho, x_\rho} < 0 \qquad f_{\iota, x_\iota}  f_{\rho, x_\rho} -  f_{\iota, x_\rho} f_{\rho, x_\iota} >  0
\end{equation}
Whether the third inequality in \eqref{e:stab_ex} holds depends on the value of the backward coupling $f_{\iota, x_\rho} = 0$.  However, whether infinitesimal homeostasis ($x_o'(0) = 0$) occurs is independent of the backward coupling since
\[
\det\Matrixc{ f_{\rho, x_\iota} &  f_{\rho, x_\rho} \\ f_{o, x_\iota} &  f_{o, x_\rho} } = 
\det\Matrixc{ f_{\rho, x_\iota} &  f_{\rho, x_\rho} \\ 0  &  f_{o, x_\rho} }  = 
f_{\rho, x_\iota} f_{o, x_\rho} = 0
\]
\end{example}

\Section{Determinant formulas}
\label{S:determinant}

Let $\PH$ be the $(n+1)\times(n+1)$ homeostasis matrix \eqref{xo'_reduced2} of the input-output network $\GG$ with input node $\iota$, $n$ regulatory nodes $\rho_j$, and output node $o$, and admissible system \eqref{eq:ad_coord}.

\begin{lemma} \label{L:summand_factors}
Every nonzero summand of $\det(\PH)$ corresponds to a unique $\iota o$-simple path and has all coupling strengths within this $\iota o$-simple path as its factors. 
\end{lemma}

\begin{proof}
Each nonzero summand in $\det(\PH)$ has $n+1$ factors and each factor is the strength of a coupling arrow or of the linearized internal dynamics of a node.  We can write $\PH$ as 
\begin{equation}\label{eq:H-expand}
\PH = \begin{gmatrix}[b]
f_{1,\iota} & f_{1,1} &  \cdots & f_{1,n-1} & f_{1,n}  \\
f_{2,\iota} & f_{2,1} &  \cdots & f_{2,n-1} & f_{2,n}\\
\vdots & \vdots & \vdots & \vdots  & \vdots\\
f_{n,\iota} & f_{n,1} & \cdots & f_{n,n-1} & f_{n,n}\\
f_{o,\iota} & f_{o,1} & \cdots & f_{o,n-1} & f_{o,n}
\colops
\mult{0}{\begin{array}{c}\mbox{$\iota$}\\ \downarrow\end{array}}
\mult{1}{\begin{array}{c}\mbox{$\rho_1$}\\ \downarrow\end{array}}
\mult{3}{\begin{array}{c}\mbox{$\rho_{n-1}$}\\ \downarrow\end{array}}
\mult{4}{\begin{array}{c}\mbox{$\rho_{n}$}\\ \downarrow\end{array}}
\rowops
\mult{0}{\leftarrow\mbox{$\rho_1$}}
\mult{1}{\leftarrow\text{$\rho_2$}}
\mult{3}{\leftarrow\text{$\rho_n$}}
\mult{4}{\leftarrow\text{$o$}}
\end{gmatrix}
\end{equation}
The columns of $\PH$ correspond to $n+1$ nodes in the order $\iota,\rho_1,\ldots,\rho_n$ and the rows of $\PH$ correspond to $n+1$ nodes in the order $\rho_1,\ldots,\rho_n,o$. The entry $f_{j,\iota} = f_{\rho_j,x_\iota}$ in column $\iota$ is the linearized coupling strength of an arrow $\iota\to\rho_j$. The entry $f_{o,k} = f_{o,x_{\rho_k}}$ in row $o$ is the linearized coupling strength of an arrow $\rho_{k}\to o$. The entry $f_{j,k} = f_{\rho_j,x_{\rho_k}}$ is the linearized coupling strength of an arrow $\rho_k\to\rho_j$.  If $j=k$, the entry $f_{k,k} = f_{\rho_k,x_{\rho_k}}$ is the linearized internal dynamics of node $k$. Note that each summand in the expansion of $\det(\PH)$ has one factor associated with each column of $\PH$ and one factor associated with each row of $\PH$. 
	
Fix a summand.  By assumption there is a unique factor associated with the first column.  If this factor is $f_{o,\iota}$, we are done and the simple path is $\iota\to o$.  So assume the factor in the first column is $f_{k,\iota}$, where $1\leq k\leq n$. This factor is associated with the arrow $\iota\to\rho_k$.
	
Next there is a unique factor in the column called $\rho_k$ and that factor corresponds to an arrow $\rho_k\to\rho_j$ for some node $\rho_j$. If node $\rho_j$ is $o$, the summand includes $(f_{k,\iota} f_{o,k})$ and the associated simple path is $\iota\to\rho_k\to o$. Hence we are done.  If not, we assume $1\leq j\leq n$. Since there is only one summand factor in each row of $\PH$, it follows that $k\neq j$. This summand is then associated with the path $\iota\to \rho_k\to\rho_j$ and contains the factors $(f_{k,\iota} f_{j, k})$.
	
Proceed inductively.  By the pigeon hole principle we eventually reach a node that connects to $o$. The simple path that is associated to the given summand is unique because we start with the unique factor in the summand that has an arrow whose tail is $\iota$ and the choice of $\rho_i$ is unique at each step.  Moreover, every coupling within this simple path is a factor of the given summand. 
\qed
\end{proof}

The determinant formula \eqref{eq:detP-form} for $\det(\PH)$ in Theorem~\ref{L:summand_form} is obtained by indexing the sum by the $\iota o$-simple paths of $\GG$ as described in Lemma~\ref{L:summand_factors}.

\begin{theorem}\label{L:summand_form}
Suppose $\GG$ has $k$ $\iota o$-simple paths $S_1,\ldots, S_k$ with corresponding complementary subnetworks $C_1,\ldots,C_k $.  Then
\begin{enumerate}[label=(\alph*)]
\item The \emph{determinant formula} holds:
\begin{equation} \label{eq:detP-form}
\det(\PH) = \sum_{i=1}^k F_{S_i}\, G_{C_i}
\end{equation}
where $F_{S_i}$ is the product of the coupling strengths within the $\iota o$-simple path $S_i$ and $G_{C_i}$ is a function of coupling strengths (including self-coupling strengths) from $C_i$.  
\item Specifically, 
\begin{equation}\label{eq:J_Ci}
G_{C_i} = \pm \det(J_{C_i})
\end{equation}
where $J_{C_i}$ is the Jacobian matrix of the admissible system corresponding to the complementary subnetwork $C_i$.  Generically, a coupling strength in $\GG$ cannot be a factor of $G_{C_i}$.
\end{enumerate}
\end{theorem}

\begin{proof}
\begin{enumerate}[label=(\alph*)]
\item  Let $S_i$ be the $r+2$ node $\iota o$-simple path $\iota\to j_1\to\cdots\to j_r\to o$ and let 
\[
F_{S_i} = f_{i_1,x_\iota} f_{i_2,x_{i_1}} \cdots f_{i_r,x_{i_{r-1}}} f_{o,x_{i_r}}
\]
be the product of all coupling strengths in $S_i$.  By Lemma~\ref{L:summand_factors}, $\det(\PH)$ has the form \eqref{eq:detP-form}.  We claim that $G_{C_i}$ is a function depending only on the coupling strengths (including self-coupling strengths) from the complementary subnetwork $C_i$.  Since each summand in the expansion of $\det(\PH)$ has only one factor in each column of $\PH$ and one factor in each row of $\PH$, the couplings in $G_{C_i}$ must have different tails and heads from the ones that appear in the simple path. Hence, $G_{C_i}$ is a function of couplings (including self-couplings) between nodes that are not in the simple path $S_i$, as claimed.  
\item Next we show that up to sign $G_{C_i}$ is the determinant of the Jacobian matrix of the admissible system for the subnetwork $C_i$ (see \eqref{eq:J_Ci}).  To this end, relabel the nodes so that the $\iota o$-simple path $S_i$ is
\[
\iota \to 1 \to \cdots \to r \to o
\]
and the nodes in the complementary subnetwork $C_i$ are labeled $r+1, \ldots, n$.  Then 
\[
F_{S_i} = (-1)^\chi f_{1, x_\iota} f_{2, x_1} \cdots f_{r,x_{r-1}} f_{o, x_r}
\]
where $\chi$ permutes the nodes of the $\iota o$-simple path $S_i$ to $1,\ldots, r$. The summands of $\det(\PH)$ associated with $S_i$ are $F_{S_i} G_{C_i}$, where 
\begin{equation} \label{e:detG}
G_{C_i} = \sum_{\sigma} (-1)^\sigma f_{r+1,x_{\sigma(r+1)}} \cdots f_{n,x_{\sigma(n)}}
\end{equation}
and $\sigma$ is a permutation of the indices $r+1,\ldots,n$. Observe that the right hand side of \eqref{e:detG} is just $\det(J_{C_i})$ up to sign.
		
Lastly, we show that no coupling strength in $\GG$ can be a factor of $\det(J_{C_i})$.  The coupling strengths correspond to the arrows and the self-coupling strengths correspond to the nodes.  The self-coupling strengths are the diagonal entries of $J_{C_i}$, which are generically nonzero.  If we set all coupling strengths to $0$ (that is, assume they are neutral), then the off-diagonal entries of $\det(J_{C_i})$ are $0$ and $\det(J_{C_i}) \neq 0$.  Now suppose that one coupling strength is a factor of $\det(J_{C_i})$, then $\det(J_{C_i}) = 0$ if that coupling is neutral and we have a contradiction.  It follows that no coupling strength can be a factor of $\det(J_{C_i})$.
\end{enumerate}
\qed
\end{proof}

\begin{theorem} \label{T:core_equivalent}
Two core networks are core equivalent if and only if they have the same set of $\iota o$-simple paths and the Jacobian matrices of the complementary subnetworks to any simple path have the same determinant up to sign.
\end{theorem}

\begin{proof}
\noindent ${\mathbf \Rightarrow}$ Let $\GG_1$ and $\GG_2$ be core networks and assume they are core equivalent.  Therefore, $\det(B_1) = \det(B_2)$ and by Theorem~\ref{L:summand_form}
\[
\det(B_1)\equiv \sum_{i=1}^k F_{S_i}\, G_{C_i} = \sum_{j=1}^\ell F_{T_j}\, G_{D_j} \equiv \det(B_2)
\] 
If a simple path of $\GG_1$ were not a simple path of $\GG_2$, the equality would fail; that is, the polynomials would be unequal.  Therefore, we may assume $\ell = k$ and (by renumbering if needed) that $T_i = S_i$ for all $i$.  It follows that 
\[
\sum_{i=1}^k F_{S_i}\, (G_{C_i} - G_{D_i}) = 0
\]
Since the $F_{S_i}$ are linearly independent it follows that $G_{C_i} = G_{D_i}$ for all $i$; that is, $\det(J_{C_i}) = \pm \det(J_{D_i})$ where $J_{C_i}$ and $J_{D_i}$ are the Jacobian matrices associated with $\GG_1$ and $\GG_2$. Hence the Jacobian matrices of the two complementary subnetworks have the same determinant up to sign.

\noindent ${\mathbf \Leftarrow}$ The converse follows directly from Theorem~\ref{L:summand_form}.
\qed
\end{proof}

\begin{corollary} \label{C:core_equivalence}
Two core networks are core equivalent if they have the same set of $\iota o$-simple paths and the same complementary subnetworks to these simple paths. 
\end{corollary}

\begin{proof}
Follows directly from Theorem \ref{T:core_equivalent}. \qed
\end{proof}

\Section{Infinitesimal homeostasis classes}
\label{S:classification2}

In this section we prove that there are two classes of infinitesimal homeostasis: {\em appendage} and {\em structural}.  See Definition~\ref{D:homeostasis_types} and Theorem~\ref{L:self_couplings}.  The section ends with a description of a `normal form' for appendage and structural homeostasis blocks.  These `normal forms' are given in Theorem~\ref{lem:associated_network}.  

Section~\ref{S:appendage_blocks} discusses graph theoretic attributes of appendage homeostasis and Section~\ref{S:structural_blocks} discusses graph theoretic attributes of structural homeostasis.  This material leads to the conclusions in Section~\ref{S:CC} where it is shown that each structural block is generated by two adjacent super-simple nodes and each appendage block is generated by a path component of the subnetwork of appendage nodes.

Recall from \eqref{eq:FK_form} that we can associate with each homeostasis matrix $H$ a set of $m$ irreducible square blocks $B_1,\ldots,B_m$ where 
\begin{equation} \label{eq:FK_form_repeat}
P \PH Q = \Matrixc{
B_1 & *   & \cdots & * \\
0   & B_2 & \cdots & * \\
\vdots &  &  & \vdots \\
0   &  0 & \cdots & B_m
}   
\end{equation}
and $P$ and $Q$ are $(n+1)\times(n+1)$ permutation matrices.  

\begin{lemma} \label{L:rows_summands}
Let $H$ be an $(n+1)\times(n+1)$ homeostasis matrix and let $P$ and $Q$ be $(n+1)\times (n+1)$ permutation matrices. Then the rows (and columns) of $PHQ$ are the same as the rows (and columns) of $H$ up to reordering. 
Moreover, the set of entries of $H$ are identical with the set of entries of $PHQ$.
\end{lemma}

\begin{proof}
The set of rows of $PH$ are identical to the set of rows of $H$.  A row of $HQ$ contains the same entries as the corresponding row of $H$---but with entries permuted.  The second statement follows from the first.
\qed
\end{proof}

Recall that the entries of the homeostasis matrix $H$, defined in \eqref{xo'_reduced2} for an admissible system of a given input-output network $\GG$, appear in three types: $0$, coupling, and self-coupling.  The following lemma is important in our discussion of homeostasis types.

\begin{lemma}
The number of self-coupling entries in each diagonal block $B_\eta$ is an invariant of the homeostasis matrix $H$. 
\end{lemma}

\begin{proof}
Suppose $H$ is transformed in two different ways to upper triangular form \eqref{eq:FK_form_repeat}. Then one obtains two sets of diagonal blocks $B_1,\ldots, B_m$ and $\tilde{B}_1,\ldots\tilde{B}_{\tilde{m}}$.  Since one set of blocks is transformed into the other by a permutation, it follows that the number of blocks in each set is the same.  Moreover, the blocks are related by
\[
\tilde{B}_{M(\nu)} = P_\nu B_\nu Q_\nu
\]  
where $M$ is a permutation of the index sets and for each $\nu$, $P_\nu$ and $Q_\nu$ are permutation matrices.  It follows from Lemma~\ref{L:rows_summands} that the size and the number of self-coupling entries of the square matrices $\tilde{B}_{M(\nu)}$ and $B_\nu$ are identical.
\qed
\end{proof}

Definition~\ref{D:classes} defined two homeostasis classes. We repeat that definition here but with more specificity.

\begin{definition} \rm \label{D:homeostasis_types}
Let $B_\eta$ be an irreducible $k\times k$ square block associated with the $(n+1)\times(n+1)$ homeostasis matrix $H$ in \eqref{eq:FK_form_repeat}.  The homeostasis class associated with $B_\eta$ is {\em appendage} if $B_\eta$ has $k$ self-coupling entries and {\em structural} if $B_\eta$ has $k-1$ self-coupling entries.  
\end{definition}

Theorem~\ref{L:self_couplings} shows that each square block is either appendage or structural. 

\begin{theorem} \label{L:self_couplings}
Let $\PH$ be an $(n+1)\times (n+1)$ homeostasis matrix and let $B_\eta$ be a $k\times k$ square diagonal block of the matrix $PHQ$ given in \eqref{eq:FK_form_repeat}, where $P$ and $Q$ are permutation matrices and $k\geq 1$.  Then  $B_\eta$ has either $k-1$ self-couplings or $k$ self-couplings.
\end{theorem}

\begin{proof}
Note that either
\begin{equation} \label{eq:PHQ}
PHQ = \Matrixc{B_\eta & D \\ 0 & E}
\quad\text{or}\quad
PHQ = \Matrixc{A & B & C \\ 0 & B_\eta & D \\ 0 & 0 & E }                                                                                     
\end{equation}
where $A$ is an nonempty square matrix. In the first case in \eqref{eq:PHQ} $B_\eta$ has single self-coupling entries in each of either $k-1$ or $k$ columns.   

We assume the second case in \eqref{eq:PHQ}.  From Lemma \ref{L:rows_summands} it follows that $PHQ$ has exactly one row and exactly one column without a self-coupling entry.  Hence, if $B_\eta$ has more than $k$ self-couplings, then $B_\eta$ and hence $\PH$ have a row with at least two self-couplings, which is not allowed. 

We show by contradiction that $B_\eta$ has at least $k-1$ self-couplings. Suppose $B_\eta$ has $\ell\leq k-2$ self-coupling entries. Note that there are $\ell$ self-couplings in $B_\eta$ by assumption, and there are no self-couplings in the 0 block. Let $b$ be the number of self-couplings in $B$.  Then $b+\ell$ is the number of self-couplings in $[B \; B_\eta \; 0]^t$.  Now, either every column or every column but one in $[B \; B_\eta \; 0]^t$ has a self-coupling.  Therefore,
\[
k - 1 \leq b+\ell \leq k
\qquad\text{or}\qquad
k- \ell - 1 \leq b \leq k-\ell
\]
We consider the two cases:
\begin{itemize}
\item Assume $b = k-\ell-1$.  Then there exists one column in $[B \; B_\eta \; 0]^t$ that has no self-couplings.  Therefore, every column in $A$ has a self-coupling. Since $B$ has a self-coupling, it follows that one row in $[A\; B\; C]$ has two self-couplings -- a contradiction. 

\item Assume $b = k -\ell$. Since $[B \; B_\eta \; 0]^t$ has self-couplings in every column, it follows that $A$ has a self-coupling in every column save at most one.  It then follows that $A$ has a self-coupling in every row save at most one.  Since $k-\ell\geq 2$, at east one row in $[A\; B\; C]$  has two self-couplings -- also a contradiction. 
\end{itemize}
Therefore, $\ell =  k-1$ or $\ell = k$.
\qed
\end{proof}

We build on Theorem~\ref{L:self_couplings} by putting $B_\eta$ into a standard form of type \eqref{K_normal_form_k}.  Its proof uses the next two lemmas about shapes and summands.   A {\em shape} $\EE$ is a subspace of $m\times n$ matrices $E=(e_{ij})$, where $e_{ij} = 0$ for some fixed subset of indices $i,j$.  A square shape $\DD$ is {\em nonsingular} if $\det(D) \neq 0$ for some $D\in\DD$.  A {\em summand} of a nonsingular shape $\DD$ is a nonzero product in $\det(D)$ for some $D\in\DD$. 

\begin{lemma}  \label{L:rows_summands_c}
The nonzero summands of $\det(PHQ)$ and $\det(H)$ are identical.
\end{lemma}

\begin{proof}
Since $\det(P) = \det(Q) = \pm 1$, it follows that $\det(PHQ) = \pm\det(H)$. Hence, the nonzero summands must be identical.
\qed
\end{proof}

\begin{lemma} \label{L:summand_products}
Suppose $\BB$ and $\CC$ are nonsingular shapes. Let $\EE$ be the shape whose size is chosen so that $\DD$ is the shape consisting of matrices
\[
D = \Matrix{B & E \\ 0 & C}
\]
where $B \in\BB$, $C\in\CC$, $E\in\EE$.
Then each summand of $\DD$ is the product of a summand of $\BB$ with a summand of $\CC$.
\end{lemma}

\begin{proof}
Suppose ${\sf d}$ is a summand of $\DD$.  The product ${\sf d}$ cannot have any entries in the $0$ block of $\DD$.  Hence, ${\sf d} = {\sf bc}$.  Moreover, there is a matrix  $B\in\BB$ such that $\det(B) = {\sf b}$ and a matrix  $C\in\CC$ such that $\det(C) = {\sf c}$.  In fact, we can assume that the nonzero entries of $B$ are precisely the entries in the nonzero product ${\sf b}$. Similarly for ${\sf c}$. Since $\det(D)\neq 0$ and $\det(D) = \det(B)\det(C)$, it follows that $\det(B) = {\sf b} \neq 0$ and $\det(C) = {\sf c}\neq 0$.  Therefore, ${\sf b}$ and ${\sf c}$ are  summands of $\BB$ and $\CC$, respectively. Conversely, assume that ${\sf b}$ and ${\sf c}$ are summands and conclude that ${\sf d}$ is also a summand.   
\qed
\end{proof}

It follows from Lemma~\ref{L:rows_summands} that the number of each type of entry in $PHQ$ is the same as the number in $H$.  Moreover, generically, the coupling and self-coupling entries are nonzero.  It follows from \eqref{xo'_reduced2} that the $n$ superdiagonal entries of $H$ are self-coupling entries and these are the only self-coupling entries in $H$.  In addition, $H$ has one self-coupling entry in each row except the last  row, and one self-coupling in each column except the first column.  By Lemma~\ref{L:rows_summands} there are exactly $n$ self-coupling entries in $PHQ$ with one in each row but one, and one in each column but one.  We use these observations in the proof of Theorem~\ref{lem:associated_network}.

\begin{theorem} \label{lem:associated_network}
Let $\PH$ be an $(n+1)\times (n+1)$ homeostasis matrix.  Suppose $\det(H)$ has a degree $k \geq 1$ irreducible factor $\det(B_\eta)$, where $B_\eta$ be a $k\times k$ block diagonal submatrix of the matrix $PHQ$ given in \eqref{eq:FK_form_repeat} and $P$ and $Q$ are permutation matrices.  If $B_\eta$  has $k-1$ self-coupling entries, then we can assume that $B_\eta$ has the form
\begin{equation} \label{eq:K_matrices(k-1)}
\Matrixc{ f_{\rho_1,x_{\rho_1}} & \cdots & f_{\rho_1,x_{\rho_{k-1}}} & f_{\rho_1,x_\ell} \\ 
\vdots & \ddots & \vdots  & \vdots \\
f_{\rho_{k-1},x_{\rho_1}} & \cdots & f_{\rho_{k-1},x_{\rho_{k-1}}} & f_{\rho_{k-1},x_\ell} \\ 
f_{j,x_{\rho_1}} & \cdots & f_{j,x_{\rho_{k-1}}} & f_{j,x_\ell}}
\end{equation}
and if $B_\eta$ has $k$ self-coupling entries, then we can assume that $B_\eta$ has the form
\begin{equation} \label{eq:K_matrices(k)}
\Matrixc{f_{\rho_1,x_{\rho_1}} & \cdots & f_{\rho_1,x_{\rho_k}} \\ 
\vdots & \ddots & \vdots \\
f_{\rho_k,x_{\rho_1}} & \cdots & f_{\rho_k,x_{\rho_k}}}
\end{equation}
\end{theorem}

\begin{proof}
Theorem \ref{L:self_couplings} implies that $B_\eta$ has either $k-1$ or $k$ self-couplings.  
Since $B_\eta$ is a $k\times k$ submatrix of $PHQ$ (a matrix that has  the same set of rows and the same set of columns as $H$), $B_\eta$ must consists of $k^2$ entries of the form
\begin{equation} \label{e:normal_K}
B_\eta = \Matrixc{f_{\rho_1,x_{\tau_1}} & \cdots & f_{\rho_1,x_{\tau_k}} \\ 
\vdots & \ddots & \vdots \\
f_{\rho_k,x_{\tau_1}} & \cdots & f_{\rho_k,x_{\tau_k}}}
\end{equation}
Since self-couplings must be in different rows and different columns we can use permutation matrices of the form  
\[
\Matrixc{I_p & 0 & 0 \\ 0 & S & 0 \\ 0 & 0 & I_q}
\]
where $S$ is a $k\times k$ permutation matrix to put $B_\eta$ in the form:
\begin{equation} \label{K_normal_form_k}
\Matrixc{ \ast & \text{sc} &\cdots &\ast\\  
\ast & \ast & \ddots & \vdots \\
\vdots & \vdots &\vdots &\text{sc} \\
\ast & \cdots & \cdots & \ast}
\quad\mbox{or} \quad
\Matrixc{ \text{sc} & \ast &\ast \\  
\ast & \ddots & \ast\\
\ast & \ast &\text{sc}}
\end{equation}
where $sc$ denotes a self-coupling entry and $\ast$ denotes either a $0$ entry or a coupling entry.  Note that we could just as well have put the self-coupling entries along the diagonal in \eqref{K_normal_form_k} (left).

If $B_\eta$ has $k-1$ self-couplings, as in \eqref{K_normal_form_k} (left), then $\rho_k \neq \tau_k$ and $\rho_j = \tau_j$ for $1\leq j \leq k-1$.  If $B_\eta$ has $k$ self-couplings, as in \eqref{K_normal_form_k} (right), then we may assume $\rho_j = \tau_j$ for all $j$. It follows that the matrices in \eqref{K_normal_form_k} have the form \eqref{eq:K_matrices(k-1)} or \eqref{eq:K_matrices(k)}.
\qed
\end{proof}

\begin{remark} \rm \label{rm:associated_network}
We use Theorem~\ref{lem:associated_network} to associate a subnetwork $\KK_\eta$ with each homeostasis $k\times k$ block $B_\eta$.  This construction implements the one in Definition~\ref{D:K_eta} for appendage and structural homeostasis blocks.   The network $\KK_\eta$ will be an input-output subnetwork with $k+1$ nodes when $B_\eta$ is structural and the network $\KK_\eta$ will be a standard subnetwork with $k$ nodes when $B_\eta$ is appendage. 

If $B_\eta$ is appendage, then the $k$ nodes in $\KK_\eta$ will correspond to the $k$ self-couplings in $B_\eta$ and the arrows in $K_\eta$ will be $\tau_i\to\tau_j$ if $h_{\tau_j,x_{\tau_j}}$ is a coupling entry in \eqref{eq:K_matrices(k)}. 

If $B_\eta$ is structural, then the $k-1$ regulatory nodes of $\KK_\eta$ will correspond to the self-couplings in $B_\eta$ and the input node $\ell$ and the output node $j$ of $\KK_\eta$ will be given by the coupling entry in \eqref{eq:K_matrices(k-1)}. The arrows in $\KK_\eta$ are given by the coupling entries of $B_\eta$.

Note that the constructions of $\KK$ from $H$ do not require that $H$ is a homeostasis block; the constructions  only require that $H$ has the form given in either \eqref{eq:K_matrices(k-1)} or \eqref{eq:K_matrices(k)}. 
\end{remark}

\Section{Appendage homeostasis blocks}
\label{S:appendage_blocks}

An appendage block $B_\eta$ has $k$ self-couplings and the form of a $k \times k$ matrix \eqref{eq:K_matrices(k)}, that is rewritten here as:
\begin{equation} \label{eq:K_matrices(k)tau}
B_\eta = \Matrixc{f_{\tau_1,x_{\tau_1}} & \cdots & f_{\tau_1,x_{\tau_k}} \\ 
\vdots & \ddots & \vdots \\
f_{\tau_k,x_{\tau_1}} & \cdots & f_{\tau_k,x_{\tau_k}}}
\end{equation}
As discussed in Remark~\ref{rm:associated_network} this homeostasis block is associated with a subnetwork $\KK_\eta$ consisting of distinct nodes $\tau_1,\ldots,\tau_k$ and arrows specified by $B_\eta$ that connect these nodes.  In this section we show that $\KK_\eta$ satisfies three additional conditions:
\begin{enumerate}[label=(\alph*)]
\item Each node $\tau_j\in\KK_\eta$ is an appendage node (Lemma~\ref{T:KK_appendage}).
\item For every $\iota o$-simple path $S$, nodes in $\KK_\eta$ do not form a cycle with nodes in $\CC_S\setminus\KK_\eta$ (Theorem~\ref{thm:appendage}(a)).
\item $\KK_\eta$ is a path component of the subnetwork of appendage nodes of $\GG$ (Theorem~\ref{thm:appendage}(b)).
\end{enumerate}

\begin{lemma} \label {C:cycle}
Suppose a nonzero summand $\beta$ of $\det(B_\eta)$ in \eqref{eq:K_matrices(k)tau} has $f_{\tau_j,{x_{\tau_i}}}$ as a factor, where $\tau_j \neq \tau_i$. Then the arrow $\tau_i \to \tau_j$ is contained in a cycle in $\KK_\eta$.
\end{lemma}

\begin{proof}
To simplify notation we drop the subscript $\eta$ below on $\tilde{H}$, $\KK$, and $\tilde\KK$.  Let $\tilde H$ be the $(k-1) \times (k-1)$ submatrix obtained by eliminating the $j^{th}$ row and the $i^{th}$ column of $B_\eta$ in \eqref{eq:K_matrices(k)tau}. Since $\tau_i \neq \tau_j$, $\tilde H$ has $k-2$ self-coupling entries.  Specifically, the two self-couplings $f_{\tau_i,x_{\tau_i}}$ and $f_{\tau_j,x_{\tau_j}}$ have been removed when creating $\tilde{H}$ from $B_\eta$.  

It follows from Remark \ref{rm:associated_network} that since $\tilde{H}$ has the form \eqref{eq:K_matrices(k-1)}, we can associate an input-output network $\tilde \KK$ with $\tilde H$, where the input node is $\tau_j$ since it does not receive any input and the output node is $\tau_i$ since it does not output to any node in $\tilde \KK$. By Lemma~\ref{L:summand_factors}, every nonzero summand in $\det(\tilde H)$ corresponds to a simple path from $\tau_j \to \tau_i$. Hence, the nonzero summand $\beta$ is given by $f_{\tau_j,x_{\tau_i}} $ times a nonzero summand corresponding to a simple path from $\tau_j \to \tau_i$. Therefore, the arrow $\tau_i \to \tau_j$ coupled with the path $\tau_j \to \tau_i$ forms a cycle in $\KK$.
\qed
\end{proof}

\begin{lemma} \label{T:KK_appendage}
Let $\KK_\eta$ be a subnetwork of $\GG$ associated with an appendage homeostasis block $B_\eta$ that consists of a subset of nodes {$\tau_1,\cdots,\tau_k$} of $\GG$. Then $B_\eta$ equals the Jacobian $J_{\KK_\eta}$ of the network $\KK_\eta$ and each node $\tau_j$ is an appendage node.
\end{lemma}

\begin{proof}
Admissible systems associated with the network $\KK_\eta$ have the form
\[
\begin{array}{rcl}
\dot{x}_{\tau_1} & = & f_{\tau_1}(x_{\tau_1},\ldots,x_{\tau_k}) \\
& \vdots & \\
\dot{x}_{\tau_k} & = & f_{\tau_k}(x_{\tau_1},\ldots,x_{\tau_k}) 
\end{array}
\]
where the variables that appear on the RHS of each equation correspond to the couplings in \eqref{eq:K_matrices(k)tau}.  It follows that the matrix $B_\eta$  in \eqref{eq:K_matrices(k)tau} equals the Jacobian $J_{\KK_\eta}$, as claimed.

We show that $\tau_j\in\KK_\eta$ is an appendage node for each $j$.  More specifically, we show that $\tau_j$ is in the complementary subnetwork $C_S$ of each $\iota o$-simple path $S$.  We now fix $\tau_j$ and $S$.

We make two claims.  First, every nonzero summand $\alpha$ of $\det(H)$ either contains the self-coupling $f_{\tau_j,{x_{\tau_j}}}$ as a factor or a coupling $f_{\tau_j,{x_{\tau_i}}}$ for some $i\neq j$ as a factor.  Second, this dichotomy is sufficient to prove the theorem.

First claim.  It follows from Lemma \ref{L:summand_products} that each summand of $\det(PHQ)$ has a summand of $\det(B_\eta)$ as a factor.  Therefore, each summand $\alpha$ of $\det(H)$ has a summand $\beta$ of $\det(B_\eta)$ as a factor. The claim follows from two facts. The first is that $B_\eta$ is the Jacobian $J_{\KK_\eta}$ and hence either the self-coupling is in $\beta$ or the off diagonal entry is in $\beta$; and the second is that once these entries are in $\beta$, they are also in $\alpha$.

Second claim. Recall that Theorem~\ref{L:summand_form} (the determinant theorem) implies that the summand $\alpha$ has the form $F_S g_{C_S}$ where $S$ is an $\iota o$-simple path,  $C_S$ is the complementary subnetwork to $S$, $F_S$ is the product of the coupling strengths within $S$, $J_{C_S}$ is the Jacobian matrix of the admissible system corresponding to  $C_S$, and $g_{C_S}$ is a summand in $\det(J_{C_S})$.

If the summand $\alpha$ has $f_{\tau_j,{x_{\tau_j}}}$ as a factor, it follows that $f_{\tau_j,{x_{\tau_j}}}$ is a factor of $g_{C_S}$ since it is a self-coupling and cannot be a factor of $F_S$. Hence, node $\tau_j$ is a node in $C_S$. 

If the summand $\alpha$ has $f_{\tau_j,{x_{\tau_i}}}$ as a factor, then $f_{\tau_j,{x_{\tau_i}}}$ is either not a factor of $F_S$ or is a factor of $F_S$. In the first case, $f_{\tau_j,{x_{\tau_i}}}$ is a factor of $g_{C_S}$. It follows that $\tau_j$ is a node in $C_S$. In the second case, the arrow $\tau_i\to\tau_j$ is on the simple path $S$. Recall that $f_{\tau_j,{x_{\tau_i}}}$ is also a factor of the summand $\beta$. It follows from Lemma~\ref{C:cycle} applied to $\beta$ that $\tau_i \to \tau_j$ is contained in a cycle in $K_\eta$. This is a contradiction since we show that $\tau_i\to\tau_j$ cannot be contained in both the simple path $S$ and a cycle in $\KK_\eta$.  

Since $\tau_i\to \tau_j$ is contained in a cycle in $\KK_\eta$, there exists an arrow $\tau_k\to \tau_i$ where $\tau_k$ is a node in $\KK_\eta$ ($\tau_k$ can be $\tau_j$). Since every nonzero summand of $\det(H)$ has a summand of $\det(B_\eta)$ as a factor, there exists a summand $F_S g_{C_S}$ having both $f_{\tau_j,x_{\tau_i}}$  and $f_{\tau_i,x_{\tau_k}}$ as factors. Note that $f_{\tau_j,x_{\tau_i}}$ is a factor of $F_S$ and $g_{C_S}$ is a summand in $\det(J_{C_S})$. Since $\tau_k\to \tau_i$ cannot be contained in $S$ it must be a factor of $g_{C_S}$. However, $C_S$ is the complementary subnetwork to $S$ that does not contain any arrow connecting to $\tau_i$ in the simple path $S$.
\qed
\end{proof}

\begin{lemma}\label{lem:J_C}
Let $\KK$ be a proper subnetwork of a subnetwork $C$ of $\GG$. If nodes in $\KK$ do not form a cycle with nodes in $C\setminus\KK$, then upon relabelling nodes $J_C$ is block lower triangular. 
\end{lemma}

\begin{proof}
The no cycle condition implies that we can partition nodes in $C$ into three classes: 
\begin{enumerate}[label=(\roman*)]
\item nodes in $C\setminus\KK$ that are strictly upstream from $\KK$, 
\item nodes in $\KK$, 
\item nodes in $C\setminus\KK$ that are not upstream from $\KK$.  
\end{enumerate}
By definition nodes in sets (i) and (iii) are disjoint from nodes in (ii).  Also, nodes in sets (i) and (iii) are disjoint because nodes in $\KK$ do not form a cycle with nodes in $C\setminus\KK$.  Finally, it is straightforward to see that $\CC = \mathrm{(i)}\cup \mathrm{(ii)}\cup \mathrm{(iii)}$. Using this partition of $\CC$, we claim that the Jacobian matrix of $C$ has the desired block lower triangular form:
\begin{equation} \label{e:J_C}
J_{C} = \left[\begin{array}{ccc|c|ccc}
\ast & \cdots &\ast & 0 & 0 & \cdots & 0\\
\vdots & \vdots & \vdots & \vdots & \vdots & \vdots & \vdots \\ 
\ast & \cdots & \ast & 0 & 0 & \cdots & 0\\ \hline
\ast & \cdots & \ast & J_{\KK} & 0 & \cdots & 0 \\ \hline
\ast & \cdots &\ast & \ast & \ast & \cdots & \ast\\
\vdots & \vdots & \vdots & \vdots & \vdots & \vdots & \vdots \\
\ast & \cdots & \ast & \ast & \ast & \cdots & \ast
\end{array}\right]
\end{equation}
Specifically, observe that there are no connections from (i) to (iii) because then a node in (iii) would be strictly upstream from $\KK$.  By definition there are no connections from (ii) to (iii).  Finally, the cycle condition implies that there are no connections from (i) to (ii). 
\qed
\end{proof}

\begin{theorem}  \label{thm:appendage}
Let $\KK_\eta$ be a subnetwork of $\GG$ associated with an appendage homeostasis block $B_\eta$. Then:
\begin{enumerate}[label=(\alph*)]
\item  For every $\iota o$-simple path $S$, nodes in $\KK_\eta$ do not form a cycle with nodes in $C_S\setminus\KK_\eta$.
\item $\KK_\eta$ is a path component of $\AA_\GG$.
\end{enumerate}
\end{theorem}

\begin{proof}
By Lemma~\ref{T:KK_appendage}, $\KK_\eta\subset \AA_\GG$ is an appendage subnetwork that is contained in each complementary subnetwork $C_S$, $B_\eta=J_{\KK_\eta}$ and $\det(J_{\KK_\eta})$ is a factor of $\det(H)$. To simplify notation in the rest of the proof, we drop the subscript $\eta$ and use $\KK$ to denote the appendage subnetwork. 

\paragraph{Proof of (a)} We proceed by contradiction and assume there is a cycle.  Let $S$ be an $\iota o$-simple path.  Let $\BB\subset C_S\setminus \KK$ be the nonempty subset of nodes that are on some cycle connecting nodes in $\KK$ with nodes in $C_S\setminus \KK$.  It follows that nodes in $\KK$ do not form any cycle with nodes in $(C_S\setminus \KK)\setminus \BB = C_S \setminus (\KK\cup\BB)$. Since $\KK\cup \BB \subset C_S$ and nodes in $\KK\cup \BB$ do not form a cycle with nodes in $C_S\setminus(\KK\cup\BB)$, by Lemma \ref{lem:J_C} we see that the Jacobian matrix of $C_S$ has the form
\begin{equation} \label{e:JCi-new}
J_{C_S} = \left[\begin{array}{ccc}
U & 0  &  0\\
\ast &  J_{\KK\cup\BB} &  0  \\ 
\ast & \ast & D
\end{array}\right]
\end{equation}
where
\[
J_{\KK\cup\BB} = \Matrixc{J_\KK & f_{\KK, x_{\BB}}\\ f_{\BB, x_{\KK}} & J_\BB}
\]
Note that $f_{\KK, x_\BB}\neq 0$ and $f_{\BB, x_\KK}\neq 0$, since there is a cycle containing nodes in $\KK$ and $\BB$.  We claim that the polynomial $\det(J_\KK)$ does not factor the polynomial $\det(J_{\KK\cup\BB})$.  It is sufficient to verify this statement for one admissible vector field. 

Relabel the nodes so that there is a cycle of nodes $1\to 2 \to \cdots \to p\to 1$ where the first $q$ nodeas are in $\KK$.  We can choose the cycle so that the remaining nodes are in $\BB$.  An admissible system for this cycle has the form 
\[
(f_1, f_2, \ldots,f_p)(x) =(f_1(x_1,x_p), f_2(x_2, x_1), \cdots, f_p(x_p,x_{p-1}))
\]
and all other coordinate functions $f_r(x) = x_r$. Hence the associated Jacobian matrix is
\begin{equation} \label{JKB}
J_{\KK\cup\BB} =  \left[\begin{array}{cccc|ccc|ccc}
f_{1,x_1} & 0 &\cdots & 0 & 0 & \cdots & f_{1,x_p}&\cdots&\cdots&0\\
f_{2,x_1}&f_{2,x_2}&\cdots & 0&0& \cdots& \cdots&\cdots&\cdots&0\\
\vdots & \ddots & \ddots & \vdots & \vdots & \vdots &\vdots&\vdots& \vdots&\vdots \\ 
\ast & \cdots & f_{q,x_{q-1}} & f_{q,x_q} & 0 & \cdots &\cdots&\cdots& \cdots&0 \\ \hline
0 & \cdots &0 & f_{q+1,x_{q}} & f_{q+1,x_{q+1}}& \cdots &\cdots&\cdots&\cdots& 0\\
\vdots & \vdots & \vdots & \vdots & \ddots & \ddots &\vdots&\vdots&\vdots& \vdots \\
0 & \cdots & 0 & 0 & 0 & f_{p,x_{p-1}} & f_{p,x_p} &\cdots&\cdots&\cdots \\\hline
0 & \cdots & 0 & 0 & 0 & \cdots & 0 &1&\cdots &0\\
\vdots & \vdots & \vdots & \vdots & \vdots & \vdots &\vdots&\vdots& \ddots& 0 \\
0 & \cdots & 0 & 0 & 0 & \cdots &\cdots& 0 &\cdots & 1
\end{array}\right]
\end{equation}
where the upper left block is $J_\KK$.
It follows from direct calculation that the determinant of the $J_{\KK\cup \BB}$ is  
\begin{equation} \label{eq:Jr}
\begin{split}
\det(J_{\KK\cup\BB}) = \; & f_{1,x_1}f_{2,x_2}\cdots f_{p,x_p} \\
& + (-1)^{(p-1)}f_{1,x_2}f_{2,x_3}\cdots f_{p,x_1}
\end{split}
\end{equation}
Hence
\[
\det(J_{\KK}) = f_{1,x_1}f_{2,x_2}\cdots f_{q,x_q}
\]
is not a factor of $\det(J_{\KK\cup\BB})$, given in \eqref{eq:Jr}.  We claim that $\det(J_\KK)$ is also not a factor of $\det(J_{C_S})$ because by \eqref{e:JCi-new}, $\det(J_{C_S}) = \det(U)\det(J_{\KK\cup\BB})\det(D)$. Suppose $ \det(J_{\KK})$ is a factor of $\det(J_{C_S})$, then it must be a factor of $\det(J_{\KK\cup\BB})$, which is a contradiction.   
Thus $\det(J_{\KK})$ is not a factor of $\det(J_{C_S})$, which contradicts the fact that $\KK_\eta$ is an appendage homeostasis block.  Hence, nodes in $\KK$ cannot form a cycle with nodes in $C_S\setminus\KK$. 

\paragraph{Proof of (b)} We begin by showing that $\KK$ is path connected; that is, there is a path from $\tau_i$ to $\tau_j$ for every pair of nodes $\tau_i, \tau_j\in\KK$.  Suppose not, then the path components of $\KK$ give $\KK$ a feedforward structure. It follows that we can partition the set of nodes in $\KK$ into two disjoint classes: $A$ and $B$ where nodes in $B$ are strictly downstream from nodes in $A$. Thus, there exist permutation matrices $P_\eta$ and $Q_\eta$ such that
\[
P_\eta B_\eta Q_\eta=\Matrix{J_A & 0\\ \ast & J_B}
\] 
which contradicts the fact that $B_\eta$ is irreducible. Therefore, $\KK$ is path connected.  

Next, we show that $\KK$ is a path component of $\AA_\GG$. Suppose that the path component $\WW\subset C_S$ of $\AA_\GG$ that contains $\KK$ is larger than $\KK$.  Then there would be a cycle in $\WW\subset C_S$ that starts and ends in $\KK$, and contains nodes not in $\KK$.  This contradicts (a) and $\WW=\KK$.  
\qed
\end{proof}


Recall from Definition \ref{D:structure_net} that $\sS_\GG$ is a subnetwork of $\GG$ that can be obtained by removing all appendage path components that satisfy the no cycle condition. 

\begin{lemma} \label{L:H'block}
Let $\GG$ be an input-output network with homeostasis matrix $H$.  Then the structural subnetwork $\sS_\GG$ is an input-output network with homeostasis matrix $H'$ and $\det(H')$ is a factor of $\det(H)$. 
\end{lemma}

\begin{proof}
By Theorem 5.2, if $B_\eta$ is an appendage homeostasis block, then the associated subnetwork $\KK_\eta$ consists of appendage nodes, and $B_\eta=J_{\KK_\eta}$. Relabel the blocks so that $B_{1},\cdots, B_p$ are appendage homeostasis blocks. We can write
\[
PHQ = \Matrixc{ J_{\KK_{1}} & \ast &\cdots &\ast\\  
0 & \ddots & \ast & \vdots \\
\vdots & 0 &J_{\KK_p} &\ast \\
0 & \cdots & 0 & H'}
\]
Hence $\det(H')$ is a factor of $\det(H)$.

Recall $H$ is an $(n+1)\times(n+1)$ matrix with $n$ self-couplings.  Since the main diagonal entries of $J_{\KK_i}$ are all self-couplings, $H'$ is a $(n+1 - \gamma) \times (n+1-\gamma)$ matrix where $\gamma$ is the total number of self-couplings in $\KK_1,\cdots, \KK_p$.  It follows that $H'$ has $n- \gamma$ self-couplings. By Theorem \ref{lem:associated_network}  we can assume $H'$ has the homeostasis matrix form and is associated with an input-output subnetwork $\sS_\GG$ of $\GG$. 

It follows from the upper triangular form of $PHQ$ that $\sS_\GG$ does not contain any node in appendage blocks or any coupling whose head or tail is a node in an appendage block. Moreover, a node that is not associated with any appendage block must be contained in $\sS_\GG$. Otherwise, the self-coupling of this node will appear in some $J_{\KK_i}$, which is a contradiction.

Hence, $\sS_\GG$ is an input-output network that consists of all nodes not associated with any appendage block and all arrows that connect nodes in $\sS_\GG$. 
\qed
\end{proof}

\begin{remark} \rm \label{R:G'}
Suppose $\KK_\eta$ is an input-output subnetwork of $\GG$ associated with an irreducible matrix $B_\eta$ in \eqref{eq:K_matrices(k-1)}. Then, it follows from Lemma~\ref{L:H'block} that $B_\eta$ is a structural block of $\GG$ if and only if $B_\eta$ is a structural block of $\sS_\GG$.
\end{remark}

\Section{Structural homeostasis blocks}
\label{S:combinatorial_blocks}

In this section we give a combinatorial description of  $\sS_\GG$ in terms of input-output subnetworks defined by super-simple nodes. We do this in four stages.

\begin{description}
\item[\S\ref{SS:ordering}] shows that the super-simple nodes in $\GG$ can be ordered by $\iota > \rho_1> \cdots >\rho_q >o$ where $a>b$ if $b$ is downstream from $a$.  See Lemma~\ref{L:supersimple}.  

\item[\S\ref{SS:adjacent}] defines the sets $\LL$ of simple nodes that lie between adjacent super-simple nodes. See Definition~\ref{D:supsimpnet} and Lemma~\ref{L:supersimpleb}.

\item[\S\ref{SS:assignment}] shows how to assign each appendage node in $\sS_\GG$ to a unique $\LL$, thus forming combinatorially the subnetwork $\LL'$. See Definition~\ref{def:L'}.

\item[\S\ref{SS:equateKtoL'}] shows that the homeostasis matrix of $\sS_\GG$ can be put in block upper trimngular form with blocks given by the homeostasis matrices of the $\LL'$.  See Corollary~\ref{cor:H'}.
\end{description}

\subsection{Ordering of super-simple nodes}
\label{SS:ordering}

\begin{lemma} \label{L:supersimple}
Super-simple nodes in $\GG$ are ordered by $\iota o$-simple paths.  
\end{lemma}

\begin{proof} 
Let  $\rho_1$ and $\rho_2$ be distinct super-simple nodes and let $S$ and $T$ be two $\iota o$-simple paths. Suppose $\rho_2$ is downstream from $\rho_1$ along $S$ and $\rho_1$ is downstream from $\rho_2$ along $T$. It follows that there is a simple path from $\iota$ to $\rho_2$ along $T$ that does not contain $\rho_1$ and a simple path from $\rho_2$ to $o$ along $S$ that does not contain $\rho_1$.  Hence, there is an $\iota o$-simple path that does not contain $\rho_1$ contradicting the fact that $\rho_1$ is super-simple. 
\qed
\end{proof}

\subsection{Simple nodes between adjacent super-simple nodes}
\label{SS:adjacent}


A super-simple subnetwork $\LL(\rho_1,\rho_2)$ is a subnetwork consisting of all simple nodes between adjacent super-simple nodes $\rho_1$ and $\rho_2$ (see Definition~\ref{D:supsimpnet}). The following Lemma shows that each non-super-simple simple node belongs to a unique $\LL$. 

\begin{lemma} \label{L:supersimpleb}
Every non-super-simple simple node lies uniquely between two adjacent super-simple nodes.  
\end{lemma}

\begin{proof} 
Let $\rho$ be a simple node that is not super-simple. By definition $\rho$ is on an $\iota o$-simple path $S$ and $\rho$ lies between two adjacent super-simple nodes $\rho_1$ and $\rho_2$ on $S$. Suppose $\rho$ is also on an $\iota o$-simple path $T$.  Then, by Lemma~\ref{L:supersimple} $\rho_1$ and $\rho_2$ must be ordered in the same way along $T$ and $\rho_1$ and $\rho_2$ must be adjacent super-simple nodes along $T$.  If $\rho$ is downstream from $\rho_2$ along $T$, then there would be an $\iota o$-simple path that does not contain $\rho_2$, which is a contradiction.  A similar comment holds if $\rho$ is upstream from $\rho_1$ along $T$. Therefore, $\rho$ is also between $\rho_1$ and $\rho_2$ on $T$.
\qed
\end{proof}

Definition~\ref{D:supsimpnet} implies that if $\rho_3$ is downstream from $\rho_2$ then
\begin{equation}\label{eq:LL-intersection}
\LL(\rho_1,\rho_2) \cap \LL(\rho_3,\rho_4) = 
\left\{ \begin{array}{cl} \emptyset & \mbox{if } \rho_3 \neq \rho_2 \\ \{\rho_2\} & \mbox{otherwise}
\end{array}\right.
\end{equation}

Lemma~\ref{L:structural} identifies several properties of the subnetworks $\LL$.

\begin{lemma} \label{L:structural}
Let the pairs of super-simple nodes $\rho_1,\rho_2$ and $\rho_3,\rho_4$ be adjacent. 
\begin{enumerate}[label=(\alph*)]
\item No arrow connects an upstream node $\rho$ in the subnetwork $\LL(\rho_1,\rho_2)$ to a downstream node $\tau$ in the  subnetwork $\LL(\rho_3,\rho_4)$ unless $\rho=\rho_2$, $\tau = \rho_3$ and $\rho_2$ and $\rho_3$ are adjacent super-simple nodes.
\item No arrow connects an upstream node $\rho$ in the subnetwork $\LL(\rho_1,\rho_2)$ to a downstream node $\tau$ in the  subnetwork $\LL(\rho_2,\rho_4)$ unless $\rho = \rho_2$ or $\tau = \rho_2$.
\item Suppose that a path of appendage nodes connects $\LL(\rho_1,\rho_2)$ to $\LL(\rho_3,\rho_4)$. Then $\rho_4$ is upstream from $\rho_1$.
\item Suppose that the appendage path component $\BB$ fails the no cycle condition and there is a cycle that connects nodes in $\BB$ with nodes in $C_S\setminus\BB$, where $C_S$ is a complementary subnetwork.  Then the nodes in $C_S\setminus\BB$ that are in the cycle are non-super-simple simple nodes that are contained in a unique super-simple subnetwork.
\end{enumerate}
\end{lemma}

\begin{proof}
\begin{enumerate}[label=(\alph*)]
\item Suppose an arrow connects a node $\rho\neq\rho_2$ in $\LL(\rho_1,\rho_2)$ to a node $\tau$ in $\LL(\rho_3,\rho_4)$ where $\rho_3$ is downstream from $\rho_2$.  Then there would be an $\iota o$-simple path that connects $\rho_1$ to $\rho$ to $\tau$ to $\rho_4$ in that order.  That $\iota o$-simple path would miss $\rho_2$, contradicting the fact that $\rho_2$ is super-simple.  A similar statement holds if $\tau\neq\rho_3$ or $\rho_2$ and $\rho_3$ are not adjacent. This proves (a).  

\item Suppose an arrow connects a node $\rho\neq\rho_2$ in $\LL(\rho_1,\rho_2)$ to a node $\tau\neq\rho_2$ in $\LL(\rho_2,\rho_4)$.  Then there would be an $\iota o$-simple path that connects $\rho_1$ to $\rho$ to $\tau$ to $\rho_4$ in that order.  That $\iota o$-simple path would miss $\rho_2$, contradicting the fact that $\rho_2$ is super-simple.  

\item Suppose $\rho_4$ is strictly downstream from $\rho_1$.  Then there is an $\iota o$-simple path from $\iota$ to $\rho_1$ to some nodes in $\AA_\GG$ to $\rho_4$ to $o$.  Therefore, at least one node in $\AA_\GG$ is not an appendage node.  A contradiction.

\item If the cycle contains a super-simple node, then the cycle cannot be in $C_S$.  Since the cycle must contain simple nodes that simple node cannot be super-simple. 

Suppose the cycle contains a simple node $\tau_1$ in $\LL(\rho_1,\rho_2)$ and another simple node $\tau_2$ in $\LL(\rho_3,\rho_4)$ where $\rho_3$ is downstream from $\rho_1$, then there would be a path connecting $\tau_1$ to $\tau_2$ that does not contain any super-simple node. This would lead to an $\iota o$-simple path from $\rho_1$ to $\tau_1$ to $\tau_2$ to $\rho_4$ that misses $\rho_2$ and $\rho_3$. Hence, the simple nodes contained in the cycle must come from a single super-simple subnetwork.
\end{enumerate}
\qed
\end{proof}

\begin{remark} \rm \label{rem:non-super-simple-simple}
Lemma~\ref{L:structural} (a,b) implies that two different super-simple subnetworks $\LL(\rho_1,\rho_2)$ and $\LL(\rho_3,\rho_4)$ where $\rho_2$ is upstream from $\rho_3$ can only be connected by  either having a common super-simple node ($\rho_2=\rho_3$) or  by having an arrow $\rho_2\to\rho_3$ where $\rho_2$ and $\rho_3$ are adjacent  super-simple nodes. 
\end{remark}

\subsection{Assignment of appendage nodes to $\LL$}
\label{SS:assignment}

By Lemma \ref{L:structural} (d) any appendage path component that fails the cycle condition forms cycles with non-super-simple simple nodes in a unique super-simple subnetwork. We can therefore expand a super-simple subnetwork $\LL$ to a super-simple structural subnetwork $\LL'$ by recruiting all appendage nodes that form cycles with nodes in $\LL$ (see Definition \ref{def:L'}). 


It follows that if $\rho_3$ is downstream from $\rho_2$, then
\begin{equation}\label{eq:LL'-intersection}
\LL'(\rho_1,\rho_2) \cap \LL'(\rho_3,\rho_4) = 
\left\{ \begin{array}{cl} \emptyset & \mbox{if } \rho_3 \neq \rho_2 \\ \{\rho_2\} & \mbox{otherwise}
\end{array}\right.
\end{equation} 
In particular, each appendage node in $\GG$ is attached to at most one $\LL$.

\begin{remark} \rm \label{rem:non-super-simple}
Suppose $\rho_3$ is downstream from $\rho_2$. By Lemma \ref{L:structural} (c) and Remark \ref{rem:non-super-simple-simple}, no arrow connects a node $\rho$ in $\LL'(\rho_1,\rho_2)\setminus\{\rho_2\}$ to a node $\tau$ in $\LL'(\rho_3,\rho_4)$ unless $\rho_2=\rho_3$ and $\tau=\rho_2$.
\end{remark}

\subsection{Relating $\sS_\GG$ with $\LL'$}
\label{SS:equateKtoL'}

\begin{proposition} \label{thm:HKK}
Let $\KK$ be an input-output core subnetwork of $\sS_\GG$ with $q+1$ super-simple nodes $\rho_1,\ldots,\rho_{q+1}$ in downstream order in $\GG$. Then the homeostasis matrix $H_\KK$ of $\KK$ can be written in an upper block triangular form 
\begin{equation} \label{eq:HKK}
H_\KK = \Matrixc{H_{\LL'_1} & \ast & \cdots & \ast\\
0 & H_{\LL'_2} & \cdots & \ast\\
\vdots&  & \ddots&\vdots\\
0 & 0& 0& H_{\LL'_q}}
\end{equation}
where for  $\ell=1,\ldots,q$, $H_{\LL'_\ell}$ is the homeostasis matrix of the super-simple structural subnetwork $\LL'_\ell = \LL'(\rho_\ell,\rho_{(\ell+1)})$. 
\end{proposition}

\begin{proof}
Since $\KK$ is an input-output core subnetwork of $\sS_\GG$, it follows that $\KK$ consists of all simple nodes between adjacent super-simple nodes of $\KK$ and appendage nodes that form cycles with non-super-simple simple nodes in $\KK$. Hence, $\KK$ consists of nodes and arrows in  $\LL'(\rho_1,\rho_2)\cup\cdots\cup\LL'(\rho_q,\rho_{q+1})$ plus backward arrows between different super-simple structural subnetworks.  Hence, for $\ell = 1,\ldots, q$, nodes in $\KK$ can be partitioned into disjoint classes: $(\ell) = \LL'_\ell \setminus \{\rho_{\ell+1}\}$.  We claim that the homeostasis matrix $H_\KK$ of $\KK$ is given by \eqref{eq:HKK}.

It follows from Remark \ref{rem:non-super-simple} that an arrow from a node in one class ($\ell$) to a node in another class (j) where $\rm j>\ell$ can exist only when the two classes are adjacent (that is,  $\rm j = \ell +1$) and the head of this arrow is the input node $\rho_{\ell+1}$ of the downstream class ($\rm \ell+1$). Since entries below $H_{\KK_\ell}$ denote the arrows from nodes in class ($\ell$) to nodes in classes ($\ell+1$) through (q) except the input node $\rho_{\ell+1}$ in class ($\ell+1$). It follows that all entries below $H_{\KK_\ell}$ are zero and hence $H_\KK$ has the upper block triangular form shown in \eqref{eq:HKK}. 
\qed
\end{proof}

\begin{corollary} \label{cor:H'}
Suppose that $\tau_1,\ldots,\tau_{p+1}$ are the super-simple nodes of $\GG$ in downstream order.   
Then the homeostasis matrix $H'$ of $\sS_\GG$ can be written in upper block triangular form 
\begin{equation} \label{eq:H'}
H' = \Matrixc{B_1' & \ast & \cdots & \ast\\
0 & B_2' & \cdots & \ast\\
\vdots& \vdots&\ddots&\vdots\\
0 & 0& 0& B_p'}
\end{equation}
where $B_\ell'$ is the homeostasis matrix of the super-simple structural subnetwork $\LL'(\tau_{\ell},\tau_{\ell+1})$ for $1\leq\ell\leq p$.  In addition, $p$ is less than or equal to the number $m$ of structural blocks $\KK_\eta$.
\end{corollary}

\begin{proof}
It follows from Definition \ref{D:structure_net} that $\sS_\GG$ has the same super-simple nodes as $\GG$ and $\sS_\GG$ is a core subnetwork. By Proposition~\ref{thm:HKK}, the homeostasis matrix $H'$ of $\sS_\GG$ is given by \eqref{eq:H'}. 
The number of irreducible blocks is the number of $\KK_\eta$ and that is $m$.  Since $m$ is the maximum number of blocks in $H'$, it follows that $m\geq p$ by \eqref{eq:H'}.  
\qed
\end{proof}

If we can show that the number of super-simple nodes in $K_\eta$ is two, then we will show that $\KK_\eta$ is core equivalent to one of the $\LL'$. 

\subsection{Relation between structural homeostasis and $\LL'$}
\label{S:structural_blocks}

This section shows that each structural subnetwork $\KK_\eta$ is core equivalent to the $\LL'$ having the same input node.   Specifically, we show that the input and output nodes in $\KK_\eta$  are adjacent super-simple and that no other nodes in $\KK_\eta$ are super-simple.

\begin{proposition} \label{T:adjacent}
Let $\KK_\eta$ be an input-output subnetwork of $\GG$ associated with an irreducible structural homeostasis matrix $B_\eta$ in \eqref{eq:K_matrices(k-1)}. Then the input and output nodes of $\KK_\eta$ are super-simple nodes.
\end{proposition}

\begin{proof}
We prove this theorem by proving that both the input and output nodes $\ell$ and $j$ of $\KK_\eta$ are on the $\iota o$-simple path associated with $\alpha$ for all summands $\alpha$ of $\det(H)$.  Theorem~\ref{L:summand_form} (the determinant theorem) implies that $\alpha$ has the form $F_S g_{C_S}$ where $S$ is an $\iota o$-simple path,  $C_S$ is the complementary subnetwork to $S$, $F_S$ is the product of the coupling strengths within $S$, $J_{C_S}$ is the Jacobian matrix of the admissible system corresponding to  $C_S$, and $g_{C_S}$ is a summand in $\det(J_{C_S})$. 

It follows from Lemma \ref{L:summand_products} that the summands of form \eqref{eq:PHQ} are the summands of $A$ times the summands of $B_\eta$ times the summands of E.  Hence, every nonzero summand of $\det(H)$ contains a nonzero summand of $\det(B_\eta)$ as a factor.  Since $\ell$ and $j$ are the input output nodes for the homeostasis matrix $B_\eta$, it follows that every nonzero summand of $\det(B_\eta)$, and hence $\det(H)$, has both $f_{m,x_\ell}$ (where $m$ is one of $\rho_1,\ldots,\rho_{k-1}, j$) and $f_{j,x_n}$ (where $n$ is one of $\rho_1,\ldots,\rho_{k-1}, \ell$) as factors. 

From the form of $PHQ$ (and hence $H$) we see that $f_{\ell,x_\ell}$ and $f_{j,x_j}$ are not factors of nonzero summands of $\det(H)$.  Suppose the summand $\alpha$ has $f_{m,x_\ell}$ as a factor, then $f_{m,x_\ell}$ is either a factor of $F_S$ or not a factor of $F_S$. In the first case, it follows that the arrow $\ell \to m$ is on the simple path $S$. Hence, the node $\ell$ is contained in $S$. In the second case, suppose $f_{m,x_\ell}$ is not a factor of $F_S$, then it must be a factor of $g_{C_S}$.  That implies that $\ell$ is a node in $C_S$. It follows that there exists another nonzero summand $\alpha'$ of $\det(H)$ which contains $f_{\ell,x_\ell}$ as a factor, which is is a contradiction. Therefore, we conclude every $\iota o$-simple path contains node $\ell$. By the same type of argument we can also conclude that every $\iota o$-simple path contains node $j$.
\qed
\end{proof}

\begin{proposition} \label{T:structural_supersimple_two}
If a structural block $B_\eta$ of $\GG$ is irreducible, then $\KK_\eta$ is an input-output subnetwork that has exactly two super-simple nodes.
\end{proposition}

\begin{proof}
By Remark~\ref{R:G'}, $\KK_\eta$ is an input-output subnetwork of $\sS_\GG$ and $\KK_\eta$ is a core subnetwork because it is irreducible. Suppose in addition to the input and output nodes there are other $q>1$ super-simple nodes in $\KK_\eta$, then by Proposition~\ref{thm:HKK}, the homeostasis matrix $B_\eta$ of $\KK_\eta$ can be written in an upper block triangular form with $q+1>2$ diagonal blocks and hence $\KK_\eta$ is reducible, a contradiction.
\qed
\end{proof}

\begin{corollary} \label{C:adjacent}
The input and output nodes of a structural homeostasis block are adjacent super-simple nodes. 
\end{corollary}

\begin{proof}
Super-simple nodes can be well-ordered.  The proof then follows from Proposition~\ref{T:structural_supersimple_two}.   \qed
\end{proof}

\begin{theorem} \label{L:KetaL'}
In $\GG$, there is a 1:1 correspondence between structural homeostasis blocks $\KK_\eta$ and super-simple structural subnetworks $\LL'$ and that correspondence is given by having the same input node.  Moreover, the corresponding $\KK_\eta$ and $\LL'$ are core equivalent. 
\end{theorem}

\begin{proof}
By Corollary~\ref{C:adjacent}, the input and output nodes of each $\KK_\eta$ are adjacent super-simple nodes and hence each $\KK_\eta$ leads to a unique $\LL'$ that has the same input node.  Therefore, the number of $\KK_\eta$ (equal to $m$) is less than or equal to the number $p$ of $\LL'$.  Corollary~\ref{cor:H'} states that $p\leq m$; hence, $p=m$. That is, there is a 1:1 correspondence between $\KK_\eta$ and $\LL'$. 

Let $\ell$ and $j$ be the input and output nodes of the structural block $\KK_\eta$. Then the corresponding super-simple structural subnetwork is $\LL'(\ell,j)$.  By Definition \ref{D:K_eta}, $\KK_\eta$ consists of simple nodes between the two adjacent super-simple nodes $\ell$ and $j$ and appendage nodes that form cycles with non-super-simple simple nodes in $\KK_\eta$.  Arrows in $\KK_\eta$ are non-backward arrows that connect nodes in $\KK_\eta$. It follows from Definition \ref{def:L'} that $\LL'(\ell, j)$ is the union of $\KK_\eta$ and arrows whose head is $\ell$ or whose tail is $j$.  By Corollary \ref{P:core_equivalent}, $\KK_\eta$ is core equivalent to $\LL'(\ell,j)$.
\qed 
\end{proof}

\Section{Classification and construction}
\label{S:CC}

In the Introduction we showed how Cramer's rule coupled with basic combinatorial matrix theory can be applied to the homeostasis matrix $H$ to determine the different types of infinitesimal homeostasis that an input-output network $\GG$ can support. Specifically the zeros of $\det(H)$, a homogeneous polynomial in the linearized couplings and self-couplings, can be factored into $\det(B_1)\cdots\det(B_m)$.  In this paper we show that there are two types of factors that depend on the number of self-couplings: one we call appendage and the other we call structural. Each factor corresponds to a type of homeostasis in subnetworks $\KK_\eta$ for $\eta = 1,\ldots,m$ that can be read directly from $\GG$. 

\paragraph{Appendage blocks}
Theorem~\ref{thm:appendage} shows that an appendage block $B_\eta$ leads to a subnetwork $\KK_\eta$ that is a path component of the appendage network $\AA_\GG\subset\GG$.  Moreover, the nodes in $\KK_\eta$ do not form a cycle with other nodes in the complementary subnetwork $\CC_S$ for every $\iota o$-simple path $S$. The factors of $\det(H)$ that stem from appendage nodes are $\det(J_\AA)$, the determinant of the Jacobian of the appendage path components $\AA$.  The converse is also valid as shown in Theorem~\ref{T:appendage_blocks}. 

\begin{theorem} \label{T:appendage_blocks}
Suppose $\KK_\eta$ is an appendage path component.  If $\KK_\eta$ satisfies the no cycle condition, then $\det(J_{\KK_\eta})$ is an irreducible factor of $\det(H)$.
\end{theorem} 

\begin{proof}
Let $C_S$ be the complementary subnetwork of an $\iota o$-simple path $S$.  By Definition \ref{D:simple_complementary}(c), $\KK_\eta\subset C_S$. Since nodes in $\KK_\eta$ do not form a cycle with other nodes in $C_S$, by Lemma \ref{lem:J_C}, $J_{C_S}$ has the following block lower triangular form:
\begin{equation} \label{e:J_C_S}
J_{C_S} = \left[\begin{array}{ccc|c|ccc}
\ast & \cdots &\ast & 0 & 0 & \cdots & 0\\
\vdots & \vdots & \vdots & \vdots & \vdots & \vdots & \vdots \\ 
\ast & \cdots & \ast & 0 & 0 & \cdots & 0\\ \hline
\ast & \cdots & \ast & J_{\KK_\eta} & 0 & \cdots & 0 \\ \hline
\ast & \cdots &\ast & \ast & \ast & \cdots & \ast\\
\vdots & \vdots & \vdots & \vdots & \vdots & \vdots & \vdots \\
\ast & \cdots & \ast & \ast & \ast & \cdots & \ast
\end{array}\right]
\end{equation}
Hence $\det(J_{\KK_\eta})$ is a factor of $\det(J_{C_S})$, and so a factor of $\det(H)$. Since $\KK_\eta$ is a path component and hence is path connected, it follows that $J_{\KK_\eta}$ is irreducible.
\qed
\end{proof}

It follows that we can  construct appendage blocks as follows.  First we determine the path components of the appendage subnetwork of $\GG$ and second we determine which of these components $\KK_\eta$ satisfy the cycle condition in Theorem~\ref{thm:appendage}.

\paragraph{Structural blocks}
Next, we form the subnetwork $\sS_\GG$ that is obtained from $\GG$ by deleting the appendage path components identified above.  The last result that is needed is:

\begin{theorem} \label{T:structural_supersimple}
Let $\ell$ and $j$ be adjacent super-simple nodes in $\sS_\GG$, then $\det(\LL'(\ell,j))$ is an irreducible factor of $\det(H)$.
\end{theorem} 

\begin{proof}
It follows from Corollary \ref{cor:H'} that $\det(\LL'(\ell,j))$ is a factor of $\det(H')$ and hence a factor of $\det(H)$ by Lemma \ref{L:H'block}. Theorem~\ref{L:KetaL'} states that $\LL'(\ell,j)$ is core equivalent to a unique $\KK_\eta$ that is irreducible. Hence, $\det(\LL'(\ell,j))$ is an irreducible factor of $\det(H)$.
\qed
\end{proof}

Next, we compute the super-simple nodes in $\sS_\GG$ in downstream order, namely, 
\[
\iota = \rho_1 > \rho_2 > \cdots > \rho_q > \rho_{q+1} = o
\]
It follows that the subnetworks $\LL'(\rho_i,\rho_{i+1})$ are core equivalent to the structural networks $\KK_\eta$.  Let $B_i$ be the homeostasis matrix associated with the input-output networks $\LL'(\rho_i,\rho_{i+1})$ and $\det(B_i)$ is a factor of $\det(H)$.

\paragraph{Acknowledgements.}
This research was supported in part by the National Science Foundation Grant DMS-1440386
to the Mathematical Biosciences Institute, Columbus, Ohio.




\end{document}